\documentclass[aps,prx,twocolumn,superscriptaddress,showpacs,final,floatfix,longbibliography]{revtex4-2}

\usepackage[utf8]{inputenc}
\usepackage[toc,page]{appendix}
\usepackage{amsfonts}
\usepackage{graphicx}
\usepackage{amsmath}
\usepackage{amssymb}
\usepackage{hyperref}
\usepackage[capitalise,nameinlink]{cleveref}
\usepackage{color}
\usepackage{mathrsfs}
\usepackage{isomath}
\usepackage{amsthm}
\usepackage{epstopdf}
\usepackage{txfonts}
\usepackage{dsfont}
\usepackage{ulem}
\usepackage{physics}
\allowdisplaybreaks[4]

\crefname{section}{Sec.}{Secs.}
\crefname{subsection}{Sec.}{Secs.}
\crefname{subsubsection}{Sec.}{Secs.}

\newtheorem{theorem}{Theorem}

\newtheorem{lemma}{Lemma}

\newtheorem{proposition}{Proposition}

\newtheorem{corollary}{Corollary}
\usepackage{bbold}
\usepackage{mathtools}
\usepackage{bm}

\newcommand{\R}{\mathbb{R}}

\renewcommand{\norm}[1]{\left\|#1\right\|}
\newcommand{\Id}{I}
\newcommand{\UPB}{\mathrm{UPB}}

\begin{document}

\title{Entanglement certification via causal-order interferometry in a quantum switch}

\author{Haojie Wang}
\address{State Key Laboratory of Artificial Microstructure and Mesoscopic Physics, School of Physics, Frontiers Science Center for Nano-optoelectronics, $\&$ Collaborative Innovation Center of Quantum Matter, Peking University, Beijing 100871, China}
\affiliation{Hefei National Laboratory, Hefei 230088, China}

\author{Shuheng Liu}
\email{liushuheng@pku.edu.cn}
\address{State Key Laboratory of Artificial Microstructure and Mesoscopic Physics, School of Physics, Frontiers Science Center for Nano-optoelectronics, $\&$ Collaborative Innovation Center of Quantum Matter, Peking University, Beijing 100871, China}

\author{Qiongyi He}
\address{State Key Laboratory of Artificial Microstructure and Mesoscopic Physics, School of Physics, Frontiers Science Center for Nano-optoelectronics, $\&$ Collaborative Innovation Center of Quantum Matter, Peking University, Beijing 100871, China}
\affiliation{Hefei National Laboratory, Hefei 230088, China}
\address{Collaborative Innovation Center of Extreme Optics, Shanxi University, Taiyuan, Shanxi 030006, China}

\date{\today}

\begin{abstract}
Entanglement certification is often performed on states that have already undergone noisy transmission or processing. Noise can reduce the surviving entanglement and can also cause a given criterion to fail even when entanglement remains. In this context, the quantum switch, a paradigmatic realization of indefinite causal order (ICO), coherently controls the orders in which two channels act and has been shown to offer advantages across a range of quantum information-processing tasks. Here we ask whether this coherent control enlarges the noise-parameter region in which entanglement remains certifiable. We regard the two order branches as the arms of a causal-order interferometer and insert a local unitary between the channel uses to tune their interference. For stochastic Pauli noise, a postselected ICO output can exhibit greater entanglement negativity than any classical mixture of the two definite orders; in particular, we identify regimes where its negativity remains nonzero while that of every classical mixture vanishes. A suitable local Pauli unitary substantially enlarges this ICO-only region, while an input-dependent path-difference indicator qualitatively links operator noncommutativity to the postselected negativity gain. Numerical examples extend the advantage to local amplitude-damping noise and two-qutrit Weyl noise. At a representative Weyl-noise point for the $3\times 3$ positive-partial-transpose (PPT) Tiles bound-entangled state, a nondecomposable witness detects the postselected ICO output, whereas an analytic bound excludes detection of the definite-order outputs and their mixtures by the entire locally rotated witness family. These results identify causal-order interferometry as a strategy for enhancing entanglement certification across distinct noise models and dimensions.
\end{abstract}

\maketitle

\section{Introduction}

Entanglement is a key resource for quantum information tasks such as quantum computation and quantum metrology \cite{GuhneEntanglementDetection2009,FriisEntanglementCertification2019}. In realistic implementations, however, the state reaching the measurement stage is affected by decoherence and other open-system noise, such as depolarization and amplitude damping, which can reduce or destroy entanglement \cite{YuFiniteTimePRL2004,AlmeidaEnvironmentInducedScience2007,AolitaOpenSystem2015}. Even when entanglement survives, noise may move the state outside the detection region of a given criterion. The Peres--Horodecki positive partial transpose (PPT) criterion is necessary and sufficient for separability in $2\times 2$ and $2\times 3$ systems \cite{PeresSeparabilityCriterionPRL1996,HorodeckiSeparabilityOfPLA1996}. Negativity and logarithmic negativity are computable partial-transpose-based entanglement monotones that vanish on all PPT states \cite{VidalComputableMeasurePRA2002,PlenioLogarithmicNegativityPRL2005}. In higher-dimensional systems, however, there exist PPT-entangled states \cite{HorodeckiMixedStatePRL1998,BennettUnextendibleProductPRL1999}, whose detection requires nondecomposable entanglement witnesses \cite{LewensteinOptimizationOfPRA2000}. Certifying the entanglement that survives such a noisy channel is therefore particularly challenging when the output lies close to or inside the PPT region.

Indefinite causal order (ICO) allows noisy channels to act in a coherent superposition of their orders. It describes processes that cannot be written as a classical mixture of definite causal orders \cite{OreshkovQuantumCorrelationsNC2012}. One typical realization is the quantum switch, where a control qubit coherently controls the order in which two operations or channels are applied \cite{ChiribellaPerfectDiscriminationPRA2012,ColnaghiQuantumComputationPLA2012,ChiribellaQuantumComputationsPRA2013}. In both theory and experiment, the quantum switch has been used to certify causal nonseparability \cite{AraujoWitnessingCausalNJP2015, Oreshkov_2016, GoswamiIndefiniteCausalPRL2018, RozemaExperimentalAspects2024, CaoSemiDevice2023, vanDerLugt2023DeviceIndependent,RichterTowardAn2026}. While originally formulated for unitary operations, the 2-slot quantum switch extends naturally to general non-unitary quantum operations \cite{AbbottCommunicationThrough2020,DongTheQuantumSwitch2023}.

Placing noisy channels in a coherent superposition of their causal orders can provide advantages in quantum information tasks that cannot be achieved with any definite causal order. The quantum switch allows two completely depolarizing channels, which preserve no input information in either definite order, to yield a nonzero classical communication capacity through the coherent interference between their two orders \cite{Ebler2018EnhancedCommunication, GoswamiIncreasingCommunicationPRResearch2020}. Coherent control of causal orders and trajectories has also been used to enhance the transmission of quantum information through noisy channels \cite{GuoExperimentalTransmissionPRL2020,RubinoExperimentalQuantumPRResearch2021,ChiribellaIndefiniteCausalNJP2021,PhysRevLett.133.040401,PhysRevA.111.012605}, and similar strategies have been applied to entanglement-related tasks such as entanglement generation and distribution, entanglement distillation, and quantum teleportation \cite{10185447, PhysRevApplied.23.054075, 1f8v-6p1d, dey, CaleffiQuantumSwitch2020, PhysRevA.108.062601, MukhopadhyaySuperpositionOf2020, BAN2023128927}. While these studies focus on generating, transmitting, or processing entanglement, we ask whether, for a fixed bipartite input and a fixed entanglement criterion, quantum-switch preprocessing enlarges the noise-parameter region in which entanglement can be certified. To infer input entanglement from a detected output, each selected conditional map must be separability preserving across $A|B$ \cite{VidalComputableMeasurePRA2002,PlenioLogarithmicNegativityPRL2005}; otherwise the preprocessing itself could create the detected entanglement.

In this work, we answer this question in the affirmative. We place two noisy channels in a quantum switch and regard the two channel orders as the arms of a causal-order interferometer. We further introduce an intermediate local unitary operation $U$ between the channels to actively control~\cite{glc7-xy8t} the interference between different causal paths. Each normalized conditional state is then compared with the outputs of the two definite orders and their arbitrary classical mixtures. For Pauli noise, we show that an inserted local Pauli unitary determines the relative sign between the two ordered Kraus products of each path pair and hence which control outcome receives that pair. Postselection can therefore route the noise components most detrimental to entanglement certification into the discarded outcome. The ICO output exhibits a larger negativity in certain noise regimes, and there exist regions where all definite-order outputs and their classical mixtures are PPT, while the ICO conditional state has a negative partial transpose (NPT). We further prove that, when all conditional maps admit separable Kraus decompositions, none of these maps can generate entanglement from separable inputs, and the sum of the probability-weighted negativities of the conditional states over all outcomes does not exceed the negativity of the input state. The Pauli and Weyl constructions below, as well as the amplitude-damping example with $U=I$, satisfy this condition.

As an example, we consider a Bell state subject to two global depolarizing channels on $A|B$ and compare the results obtained without an intermediate local unitary and with the insertion of $U=\sigma_z\otimes \sigma_z$. In the case of  $U=\sigma_z\otimes \sigma_z$, the optimal postselected ICO output exhibits a substantially larger negativity than the definite-order outputs, while entanglement remains certifiable even as one channel approaches the completely depolarizing limit, where all definite-order outputs and their classical mixtures are PPT. We also introduce the difference between the paths associated with the two causal orders to characterize how operator noncommutativity \cite{3jlc-lb5c, kong2026enhancedquantummetrologycriticalityassisted, PhysRevLett.130.170201} enters causal-order interference and how it relates to the ICO advantage. When no operation is applied between the channels, this path difference reduces to the commutator of the corresponding Kraus operators. We also find a positive postselected negativity gain for local amplitude-damping noise, and for two-qutrit Weyl noise the ICO protocol again certifies NPT entanglement over a broader parameter region; the certification advantage therefore persists under non-Pauli noise and in high-dimensional systems. For PPT entanglement, taking the $3\times 3$ Tiles bound-entangled state constructed from an unextendible product basis (UPB) \cite{BennettUnextendibleProductPRL1999} as an example, we derive an analytic lower bound that validates an explicit nondecomposable entanglement witness. At the representative parameters analyzed in \cref{subsec5_C}, we prove that no witness in this locally rotated family detects the DCO outputs and identify an explicit member that detects the ICO conditional output.

\section{Quantum switch as a causal-order interferometer}\label{sec2}

We consider a bipartite discrete-variable state $\rho_{AB}$ and compare two preprocessing strategies before applying the same entanglement criterion across $A|B$. In both strategies, $\rho_{AB}$ passes through two completely positive trace-preserving channels $\mathcal E_1$ and $\mathcal E_2$, either in a definite order or in a quantum switch, with Kraus representations
\begin{equation}
\mathcal{E}_{1}(\rho)=\sum_{i}K_{i}^{(1)}\rho K_{i}^{(1)\dagger},\quad
\mathcal{E}_{2}(\rho)=\sum_{j}K_{j}^{(2)}\rho K_{j}^{(2)\dagger},
\end{equation}
where $K^{(1)}_i$ and $K^{(2)}_j$ are the Kraus operators of $\mathcal{E}_1$ and $\mathcal{E}_2$, respectively, satisfying $\sum_iK_i^{(1)\dagger}K_i^{(1)}=I$ and $\sum_jK_j^{(2)\dagger}K_j^{(2)}=I$.
\begin{figure}[htpb]
\centering
\includegraphics[width=0.45\textwidth]{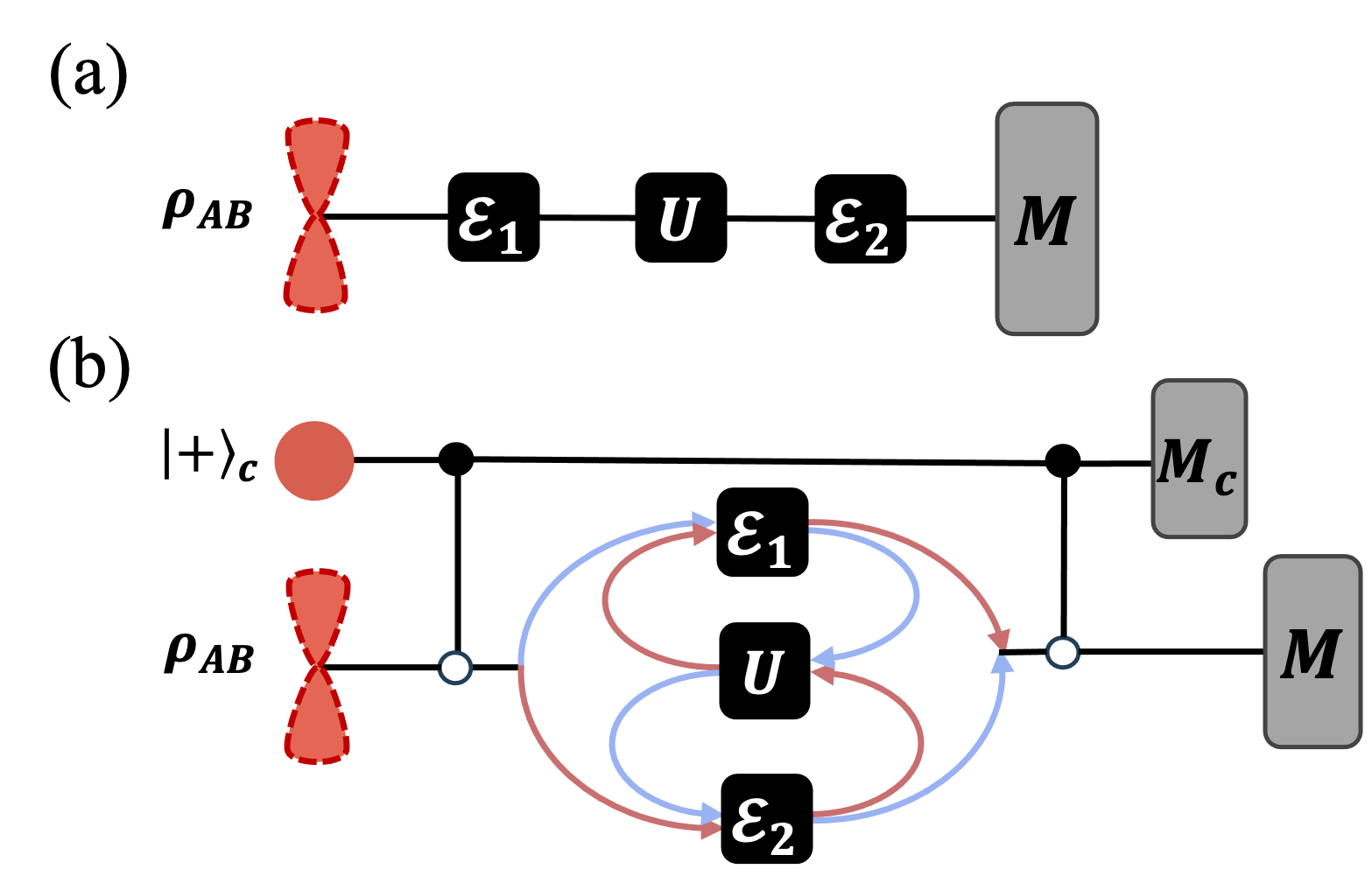}
\caption{
Schematic of the preprocessing strategies compared for entanglement certification. (a) In the definite causal order (DCO) benchmark, the target state undergoes the two noisy channels sequentially, with a local unitary inserted between them. (b) In the indefinite causal order (ICO) protocol, a control qubit coherently controls the two channel orders. Measuring the control selects a conditional target output, on which the same entanglement measurement is then performed.
}
\label{fig:1}
\end{figure}

Throughout this work, each noisy channel is assumed to admit a Kraus representation consisting of product Kraus operators with respect to $A|B$, so neither channel generates entanglement from an $A|B$-separable input. A local unitary $U=U_A\otimes U_B$ may be applied once between the first and second channel uses in either order component to tune their interference. \cref{fig:1} summarizes the definite causal order (DCO) and ICO preprocessing structures compared in this paper. For notational convenience, we use $\rightarrow$ to denote the order $\mathcal{E}_1$ followed by $U$ and $\mathcal{E}_2$, while $\leftarrow$ denotes the reversed order. The definite-order Kraus operators are $H_{ij}^{\rightarrow}=K_j^{(2)} U K_i^{(1)}$ and $H_{ij}^{\leftarrow}=K_i^{(1)} U K_j^{(2)}$. The corresponding DCO outputs are
\begin{equation}
\rho_f^{\rightarrow}=\sum_{i,j}H_{ij}^{\rightarrow}\rho_{AB}
(H_{ij}^{\rightarrow})^\dagger,
\qquad \rho_f^{\leftarrow}=\sum_{i,j}H_{ij}^{\leftarrow}\rho_{AB}
(H_{ij}^{\leftarrow})^\dagger ,
\end{equation}
and their classical mixture is
\begin{equation}
\rho_f^{\rm MIX}(\xi)=\xi\rho_f^{\rightarrow}+(1-\xi)\rho_f^{\leftarrow}
\label{eq:mix}
\end{equation}
with $\xi\in [0,1]$. When $\xi=1/2$, we omit the argument and write the equally weighted mixture simply as $\rho_f^{\rm MIX}$. The classical benchmark $\rho_f^{\rm MIX}(\xi)$ is a convex mixture of the two definite-order outputs and therefore retains no coherence between the order branches. In the quantum-switch protocol shown in \cref{fig:1}(b), a control qubit is initialized in $|+\rangle_c=(|0\rangle+|1\rangle)/\sqrt{2}$, such that the control states $|0\rangle_c$ and $|1\rangle_c$ are associated with the ordered transformations $H_{ij}^{\rightarrow}$ and $H_{ij}^{\leftarrow}$, respectively. The joint Kraus operators are \cite{ChiribellaQuantumComputationsPRA2013,DongTheQuantumSwitch2023}
\begin{equation}
K_{ij}^{\rm SW}
= |0\rangle\langle 0|\otimes H_{ij}^{\rightarrow}
+ |1\rangle\langle 1|\otimes H_{ij}^{\leftarrow}.
\end{equation}
Measuring the control in the phase basis $|\pm_{\phi}\rangle=(|0\rangle\pm e^{i\phi}|1\rangle)/\sqrt{2}$ induces the conditional Kraus operators
\begin{equation}
M_{ij}^{\pm,\phi}
= \frac{1}{2}\left(H_{ij}^{\rightarrow}
\pm e^{-i\phi}H_{ij}^{\leftarrow}\right).
\label{eq:M_ij}
\end{equation}
Control postselection coherently combines the definite-order Kraus operators $H_{ij}^{\rightarrow}$ and $H_{ij}^{\leftarrow}$, whereas the classical benchmark in \cref{eq:mix} combines only the resulting output states. The corresponding trace-nonincreasing completely positive maps are
\begin{equation}
\Lambda_{\pm,\phi}(\rho)
= \sum_{i,j}M_{ij}^{\pm,\phi}\rho M_{ij}^{\pm,\phi\dagger},
\label{eq:Lambda}
\end{equation}
where $\pm$ labels the control measurement outcome and $\phi$ specifies the measurement phase. For input $\rho_{AB}$, let $\rho_f^{\pm,\phi}=\Lambda_{\pm,\phi}(\rho_{AB})$ and $P_{\pm,\phi}=\operatorname{Tr}\rho_f^{\pm,\phi}$. If $P_{\pm,\phi}>0$, the normalized conditional state is $\widehat{\rho}_f^{\pm,\phi}=\rho_f^{\pm,\phi}/P_{\pm,\phi}$. Combining \cref{eq:M_ij,eq:Lambda}, we have
\begin{equation}
\begin{aligned}
\rho_f^{\pm,\phi}
&= \frac{1}{4}\sum_{i,j}(
H_{ij}^{\rightarrow}\rho_{AB}
(H_{ij}^{\rightarrow})^\dagger
+ H_{ij}^{\leftarrow}\rho_{AB}
(H_{ij}^{\leftarrow})^\dagger\\
&\quad \pm e^{i\phi} H_{ij}^{\rightarrow}\rho_{AB}
(H_{ij}^{\leftarrow})^\dagger
\pm e^{-i\phi} H_{ij}^{\leftarrow}\rho_{AB}
(H_{ij}^{\rightarrow})^\dagger ).
\end{aligned}
\label{eq:final_state}
\end{equation}
The last two terms are the cross-order interference contributions.

The enhancement due to causal-order interference is related to the operator noncommutativity along the ordered paths. For an $X$-basis control measurement, i.e., $\phi=0$, the interference terms in \cref{eq:final_state} motivate the path-difference operator
\begin{equation}
\Delta_{ij}(U)=H_{ij}^{\rightarrow}-H_{ij}^{\leftarrow}.
\end{equation}
Equivalently, this operator can be expressed entirely in terms of commutators as
\begin{equation}
\begin{aligned}
\Delta_{ij}(U)
={}&K_j^{(2)}[U,K_i^{(1)}]
+[K_j^{(2)},K_i^{(1)}]U\\
&+K_i^{(1)}[K_j^{(2)},U].
\end{aligned}
\end{equation}
This decomposition makes explicit how operator noncommutativity differentiates the two ordered paths~\cite{PhysRevLett.130.170201}. When $U=\mathbb I$, it reduces to $\Delta_{ij}=[K_j^{(2)},K_i^{(1)}]$. We then construct an operator associated with the input state and all pairs of Kraus operators,
\begin{equation}
\mathcal{Q}_\rho(U) = \sum_{i,j} \Delta_{ij}(U)\rho_{AB}
\Delta_{ij}^\dagger(U).
\label{eq:Q-rho}
\end{equation}
Since $M_{ij}^{-,0}=\Delta_{ij}(U)/2$, the unnormalized $(-)$ conditional output for the $X$-basis measurement satisfies
\begin{equation}
\rho_f^{-,0}=\frac{1}{4}\mathcal Q_\rho(U).
\label{eq:conditional-Q}
\end{equation}
The individual operators $\Delta_{ij}(U)$ provide a pathwise decomposition and depend on the chosen Kraus representations. Their aggregate $\mathcal Q_\rho(U)=4\rho_f^{-,0}$, however, is fixed by the physical conditional output and is therefore representation independent. We use the Hilbert--Schmidt norm
\begin{equation}
\mathcal{I}(U) = \|\mathcal{Q}_\rho(U)\|_{\mathrm{HS}}
\end{equation}
as an input-dependent scalar indicator of the path-difference magnitude, and compare it with the negativity gain of the normalized $(+)$-conditional state in \cref{non_commu}.

\section{Pauli channels}

Causal-order interference is most explicit for Pauli channels with a local Pauli unitary inserted between them, which we therefore analyze in detail below.

\subsection{Exact selection of Pauli paths by causal-order interference}

For Pauli Kraus operators and a Pauli-string insertion $U$, the two ordered products $H_{ij}^{\rightarrow}$ and $H_{ij}^{\leftarrow}$ associated with a path pair contain the same Pauli string and therefore differ at most by a sign. An $X$-basis measurement of the control consequently assigns each path pair exclusively to one outcome according to this relative sign~\cite{GoswamiIncreasingCommunicationPRResearch2020, Kechrimparis2025ProbabilisticDistillation, 10048485}, as formalized in Theorem~\ref{thm:pauli-parity}.
\begin{theorem}\label{thm:pauli-parity}
Consider two noisy channels with Pauli Kraus realizations $K_i^{(1)}=\sqrt{p_i}P_i$ and $K_j^{(2)}=\sqrt{q_j}Q_j$, where $P_i$ and $Q_j$ are Pauli strings, and let the local unitary operation $U$ also be a Pauli string. Define the relative sign label $s_{ij}(U)\in \{0,1\}$ through $Q_j U P_i = (-1)^{s_{ij}(U)} P_i U Q_j$. For $\phi=0$, the ICO conditional maps satisfy
\begin{equation}
\Lambda_+(\rho)
= \sum_{s_{ij}(U)=0}
p_i q_j
S_{ij}\rho S_{ij}^\dagger,
\qquad \Lambda_-(\rho)
= \sum_{s_{ij}(U)=1}
p_i q_j
S_{ij}\rho S_{ij}^\dagger,
\end{equation}
where $S_{ij}=P_i U Q_j$. Thus, the $(+)$ and $(-)$ control outcomes retain exactly the path pairs with $s_{ij}(U)=0$ and $s_{ij}(U)=1$, respectively.
\end{theorem}
\begin{proof}
For each path pair, $H_{ij}^{\rightarrow} = \sqrt{p_i q_j} Q_j U P_i$ and $H_{ij}^{\leftarrow} = \sqrt{p_i q_j} P_i U Q_j$. Since Pauli strings either commute or anticommute, $s_{ij}(U)$ is well defined and $H_{ij}^{\rightarrow} = (-1)^{s_{ij}(U)} H_{ij}^{\leftarrow}$. Substituting this relation into \cref{eq:M_ij} at $\phi=0$ gives
\begin{equation}
M_{ij}^{+}
= \frac{1+(-1)^{s_{ij}(U)}}{2}
H_{ij}^{\leftarrow},\qquad
M_{ij}^{-}
= \frac{(-1)^{s_{ij}(U)}-1}{2}
H_{ij}^{\leftarrow}.
\end{equation}
For $s_{ij}(U)=0$, $M_{ij}^{-}=0$ and the path pair contributes only to the $(+)$ outcome. For $s_{ij}(U)=1$, $M_{ij}^{+}=0$ and $M_{ij}^{-}=-H_{ij}^{\leftarrow}$, where the sign is an irrelevant global phase that cancels in \cref{eq:Lambda}.
\end{proof}

For a local Pauli unitary $U$, every nonvanishing conditional Kraus operator is proportional to $S_{ij}=P_iUQ_j$, which is a product operator with respect to $A|B$. Each conditional map therefore admits a separable Kraus decomposition and cannot create $A|B$ entanglement from a separable input; \cref{Lemma1} below formalizes this statement.

Theorem~\ref{thm:pauli-parity} characterizes the causal-order interference at the level of individual Pauli paths. Because $s_{ij}(U)$ depends on $U$, the inserted Pauli unitary determines which control outcome retains each path pair, so postselection can route the more harmful Pauli noise components into the discarded outcome. This algebraic feature has appeared in studies of communication enhancement \cite{Ebler2018EnhancedCommunication, GoswamiIncreasingCommunicationPRResearch2020, RubinoExperimentalQuantumPRResearch2021, PhysRevA.111.012605}, perfect communication \cite{ChiribellaIndefiniteCausalNJP2021}, and probabilistic channel distillation of Pauli channels \cite{Kechrimparis2025ProbabilisticDistillation}.

\subsection{NPT certification of two-qubit target states}\label{subsec2-C}

We take $\rho_{AB}=|\Phi^+\rangle\langle\Phi^+|$ with $|\Phi^+\rangle=(|00\rangle+|11\rangle)/\sqrt{2}$, and let both $\mathcal E_1$ and $\mathcal E_2$ be global depolarizing channels acting on the four-dimensional joint system $AB$. Each such global depolarizing channel has the Pauli-random-unitary form
\begin{equation}
\mathcal E_{\mathrm{DPC}}^{(2q)}(\rho)
= p\rho+\frac{1-p}{15}
\sum_{P\in \mathcal P_2\setminus\{I\otimes I\}}P\rho P^\dagger,
\end{equation}
where $p=(15\lambda+1)/16$ and $\lambda\in [0,1]$ is the depolarizing parameter of the standard two-qubit depolarizing channel, and $\mathcal P_2=\{P_A\otimes P_B: P_A,P_B\in \{I,X,Y,Z\}\}$. Let $\boldsymbol{\lambda}=(\lambda_1,\lambda_2)$ denote the depolarizing parameters of the two channels, and let $U$ be a fixed local unitary operation. We consider control measurements in the $X$ basis, i.e., the phase basis with $\phi=0$, so each run yields one of the two outcomes $\pm$. When negativity $\mathcal{N}(\rho_{AB}) =(\|\rho_{AB}^{T_B}\|_1-1)/2 $ is used to quantify entanglement, its convexity follows from the linearity of partial transposition and the convexity of the trace norm. Hence, no classical mixture can exceed the larger negativity of the two endpoints. According to \cref{eq:mix}, the classical causal order benchmark can be written equivalently as
\begin{equation}
\mathcal N_{\mathrm{cl}}^{U}(\boldsymbol{\lambda})
= \max_{\xi\in [0,1]}
\mathcal N(\rho_{f,\boldsymbol{\lambda}}^{{\rm MIX},U}(\xi))
= \max
\{ \mathcal N(\rho_{f,\boldsymbol{\lambda}}^{\rightarrow,U}),
\mathcal N(\rho_{f,\boldsymbol{\lambda}}^{\leftarrow,U}) \}.
\label{eq:classical-negativity-comparison}
\end{equation}
If $\mathcal N_{\mathrm{cl}}^{U}(\boldsymbol{\lambda})=0$, then the PPT criterion detects none of these outputs as entangled. With $r\in \{(+,0),(-,0)\}$ denoting the two control outcomes at $\phi=0$, we define the optimized negativity of the ICO conditional state for the same inserted unitary as
\begin{equation}
\mathcal N_{\mathrm{ICO}}^{U}(\boldsymbol{\lambda})
= \max_{r}
\mathcal N(\widehat\rho_{f,\boldsymbol{\lambda}}^{r,U}).
\end{equation}
The negativity gain of the conditional state and the ICO-only NPT-certifiable region are then defined as
\begin{equation}
\Delta_{\mathcal N}^{U}(\boldsymbol{\lambda})
= \mathcal N_{\mathrm{ICO}}^{U}(\boldsymbol{\lambda})-\mathcal N_{\mathrm{cl}}^{U}(\boldsymbol{\lambda}),
\end{equation}
\begin{equation}
\mathcal R_{\mathrm{ICO}\text{-}\mathrm{only}}^{U}
= \{ \boldsymbol{\lambda}:
\mathcal N_{\mathrm{cl}}^{U}(\boldsymbol{\lambda})=0,
\mathcal N_{\mathrm{ICO}}^{U}(\boldsymbol{\lambda})>0,
P_{\rm opt}^{U}(\boldsymbol{\lambda})>0 \}.
\end{equation}
Here, $P_{\rm opt}^{U}(\boldsymbol{\lambda})$ is the success probability of the optimal control outcome that attains $\mathcal N_{\mathrm{ICO}}^{U}$, while $\mathcal N_{\mathrm{ICO}}^{U}$ itself is the negativity of the normalized conditional state. Thus, $\mathcal R_{\mathrm{ICO}\text{-}\mathrm{only}}^{U}$ is the set of noise parameters for which every definite-order output and every classical mixture of them is PPT, while at least one nonzero-probability ICO branch is NPT. As shown in Theorem~\ref{thm:pauli-parity}, the choice of $U$ can change the sign label $s_{ij}(U)$ assigned to each path pair. \cref{fig:pauli} compares $U=I$ and $U=\sigma_z\otimes \sigma_z$ in the noise parameter space. For $U=I$, the negativity gain is positive over a broad region but remains small, and the ICO-only NPT-certifiable region is narrow. With $U=\sigma_z\otimes \sigma_z$, this region expands substantially, the peak negativity gain is approximately one order of magnitude larger, and the selected conditional output remains NPT as one channel approaches complete depolarization.
\begin{figure}[htbp]
\centering
\includegraphics[width=0.5\textwidth]{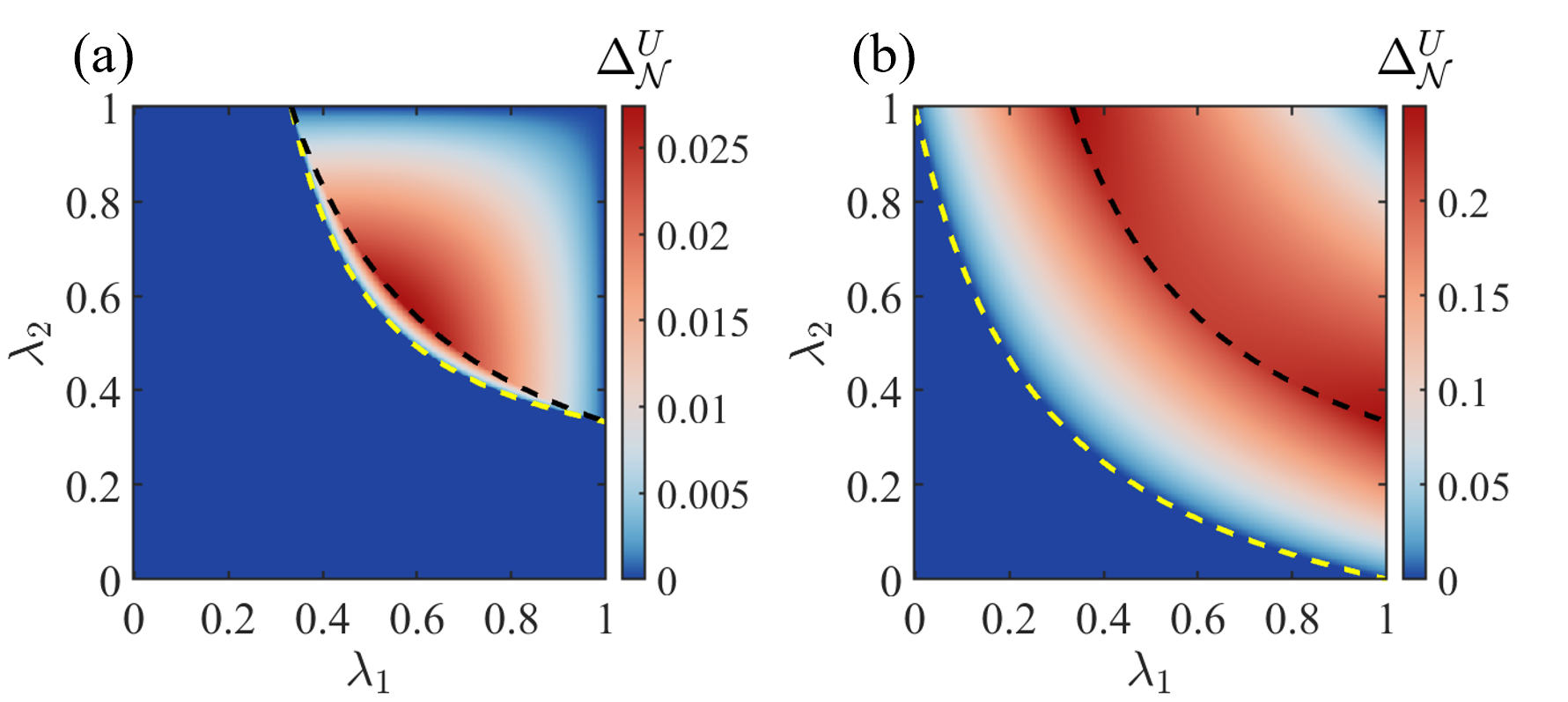}
\caption{Conditional-state negativity gain $\Delta_{\mathcal N}^{U}$ and NPT-certifiable regions under two global depolarizing channels on $AB$. (a) $U=I$. (b) $U=\sigma_z\otimes \sigma_z$. The horizontal and vertical axes denote the depolarizing parameters $\lambda_1$ and $\lambda_2$ of the two channels, and the color scale gives $\Delta_{\mathcal N}^{U}$. The yellow and black dashed curves mark the ICO and DCO NPT boundaries, respectively; only the ICO protocol certifies NPT entanglement in the region between them.}
\label{fig:pauli}
\end{figure}

We next state a sufficient condition under which a postselected branch cannot create $A|B$ entanglement from a separable input. The proof is given in \cref{proof_L1}.
\begin{lemma}
\label{Lemma1}
Let $A|B$ be the bipartition to be certified. Let $\{\Lambda_r\}_r$ be a set of trace-nonincreasing completely positive maps labeled by a classical outcome $r$, such that the total map $\Lambda=\sum_r\Lambda_r$ is trace preserving. For an input state $\rho_{AB}$, we define the outcome probability and the normalized conditional state as $p_r=\operatorname{Tr}\Lambda_r(\rho_{AB})$ and $\rho_{AB}^{(r)}=\Lambda_r(\rho_{AB})/p_r$, considering only outcomes with $p_r>0$.

If each $\Lambda_r$ admits a separable Kraus decomposition with respect to $A|B$, then, for any separable input $\sigma_{AB}\in {\rm SEP}$, the conditional output $\sigma_{AB}^{(r)}$ remains in ${\rm SEP}$. Hence, the conditional map $\Lambda_r$ cannot generate $A|B$ entanglement from a separable input.

More generally, suppose that each branch is PPT-preserving, meaning that $\Phi_r=\Gamma_B\circ \Lambda_r\circ \Gamma_B$ is positive for every outcome $r$, where $\Gamma_B$ denotes partial transposition on subsystem $B$. Then the probability-weighted negativity of the conditional states obeys the average monotonicity relation~\cite{VidalComputableMeasurePRA2002,PlenioLogarithmicNegativityPRL2005}
\begin{equation}
\sum_r p_r\mathcal N(\rho_{AB}^{(r)})
\le \mathcal N(\rho_{AB}).
\end{equation}
In particular, for each outcome $r$,
\begin{equation}
p_r\mathcal N(\rho_{AB}^{(r)})
\le \mathcal N(\rho_{AB}).
\end{equation}
\end{lemma}
For the Pauli conditional maps considered above, Lemma~\ref{Lemma1} also gives the average bound $\sum_r p_r\mathcal N(\rho_{AB}^{(r)})\le \mathcal N(\rho_{AB})$; the reported gain therefore arises from selection among the control outcomes of these separability-preserving conditional maps rather than from the creation of entanglement from a separable input. \cref{app_choi} gives a channel-level example in the Choi representation \cite{CHOI1975285, JAMIOLKOWSKI1972275} for single-qubit channels, and \cref{App_robust} examines inputs with different Schmidt coefficients.

\subsection{Postselection success probability}

The postselection success probability quantifies the resource cost of postselection in the ICO protocol: accepting an outcome of probability $P$ requires $1/P$ independent runs on average. Throughout this section we use the trace-nonincreasing conditional map $\Lambda_{\pm,\phi}$, the unnormalized conditional state $\rho_f^{\pm,\phi}$, and the success probability $P_{\pm,\phi}$ introduced in \cref{sec2}.

According to the definitions in \cref{subsec2-C}, \cref{fig:probabilities} shows the phase diagram of the success probability $P^U_{\rm opt}(\boldsymbol{\lambda})$ of the outcome that attains $\mathcal N_{\mathrm{ICO}}^{U}$. Throughout the scans reported in this section, the maximum is attained by the $(+)$ outcome, so $P^U_{\rm opt}=P_{+,0}$. For two global depolarizing channels on $AB$, $P^U_{\rm opt}$ is no smaller than $1/2$ over the whole scanned parameter region, both for $U=I$ and for an inserted $U=\sigma_z\otimes \sigma_z$. In the asymmetric regime shown in \cref{fig:probabilities}(b), where one channel approaches the completely depolarizing limit while the other remains weakly noisy, the success probability approaches its minimum value of $1/2$ for obtaining the optimal conditional state. This regime overlaps with the ICO-only region, where the selected conditional output remains NPT even under near-extreme noise. This observation indicates that the extension of the certifiable entanglement region of the conditional state comes at the cost of a reduced postselection probability, compatible with the probability-weighted monotonicity in \cref{Lemma1}.
\begin{figure}[htbp]
\centering
\includegraphics[width=0.5\textwidth]{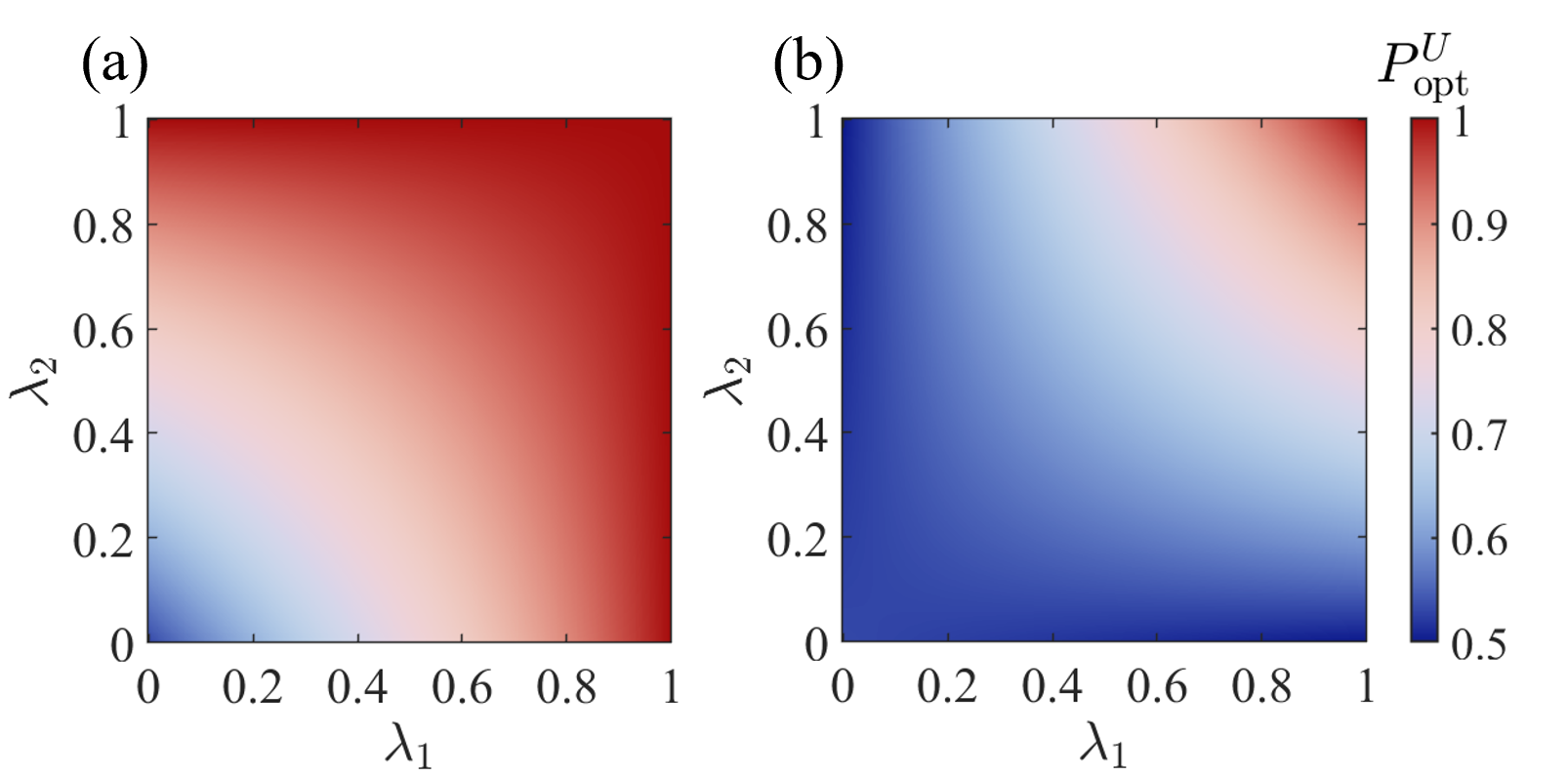}
\caption{
Postselection probability for the optimal conditional output state under two global depolarizing channels on $AB$. The probability is plotted as a function of the channel parameters $\lambda_1$ and $\lambda_2$. Panels (a) and (b) correspond to the cases without local operations and with the local operation $\sigma_z\otimes \sigma_z$, respectively. The blue-to-red color scale indicates an increasing probability for the control outcome that attains $\mathcal N_{\mathrm{ICO}}^{U}$.}
\label{fig:probabilities}
\end{figure}

The success probabilities in \cref{fig:probabilities} follow directly from the $X$-basis decomposition. \cref{eq:Q-rho} and \cref{eq:conditional-Q} give $P_{-,0}=\operatorname{Tr}[Q_{\rho}(U)]/4$ and $P_{+,0}=1-\operatorname{Tr}[Q_{\rho}(U)]/4$, so $\operatorname{Tr}[Q_{\rho}(U)]$ determines the probability distribution between the two conditional branches. By contrast, $\mathcal{I}(U)=\lVert Q_{\rho}(U)\rVert_{\mathrm{HS}}$ quantifies the overall magnitude of the path difference contribution. These quantities therefore characterize complementary aspects of causal-order interference: its postselection statistics and its redistribution of contributions between the conditional outputs. This distinction underlies the analysis below of the relation between path difference and conditional-state negativity under general local unitaries.

\subsection{General local unitaries and path difference under depolarizing noise}

For general local unitaries, the two ordered products associated with a Pauli path need not differ only by a sign, and the conditional ICO Kraus operators need not remain product operators. Because separability preservation is then not guaranteed, the scans below compare output-state negativities rather than certify input entanglement, with $\mathcal I(U)$ used only as an input-dependent path-difference indicator.

\subsubsection{Negativity under general local unitary operations}

We parametrize $U=U^{(A)}\otimes U^{(B)}$, where $U^{(A)}$ and $U^{(B)}$ are two single-qubit operations, namely $U^{(k)}=\exp(-i\alpha^{(k)}\hat{n}_k\cdot \hat{\sigma})$, $k\in \{A,B\}$. Each single-qubit operation can be expanded in the Pauli basis,
\begin{equation}
U^{(k)}=u_{I}^{(k)}I+u_{X}^{(k)}X+u_{Y}^{(k)}Y+u_{Z}^{(k)}Z,
\end{equation}
where
\begin{equation}
\begin{aligned}
u_{I}^{(k)}&= \cos(\alpha^{(k)}),\\
u_{X}^{(k)}&= -i\sin(\alpha^{(k)})\sin(\theta^{(k)})\cos(\phi^{(k)}),\\
u_{Y}^{(k)}&= -i\sin(\alpha^{(k)})\sin(\theta^{(k)})\sin(\phi^{(k)}),\\
u_{Z}^{(k)}&= -i\sin(\alpha^{(k)})\cos(\theta^{(k)}),
\end{aligned}
\end{equation}
with $\alpha^{(k)}\in [0,\pi],\theta^{(k)}\in [0,\pi],\phi^{(k)}\in [0,2\pi]$ and $k\in \{A,B\}$. For simplicity, the two channels have the same identity Kraus weight $p_1=p_2=p$, corresponding to $\lambda_1=\lambda_2=(16p-1)/15$, and we set $\alpha^{(A)}=\alpha^{(B)}=\alpha$, $\theta^{(A)}=\theta^{(B)}=\theta$, and $\phi^{(A)}=\phi^{(B)}=\phi_0$. In each scan, we fix $p=0.8$ and two angular parameters, and vary the remaining parameter to examine the dependence of the conditional-state negativity. Since the depolarizing channel is unitarily covariant, the inserted local unitary can be commuted through the subsequent channel in either definite order and then acts as a local rotation on the output. This leaves the negativity across $A|B$ unchanged, so the negativity of either definite-order output does not depend on the inserted local unitary.
\begin{figure}[htbp]
\centering
\includegraphics[width=0.48\textwidth]{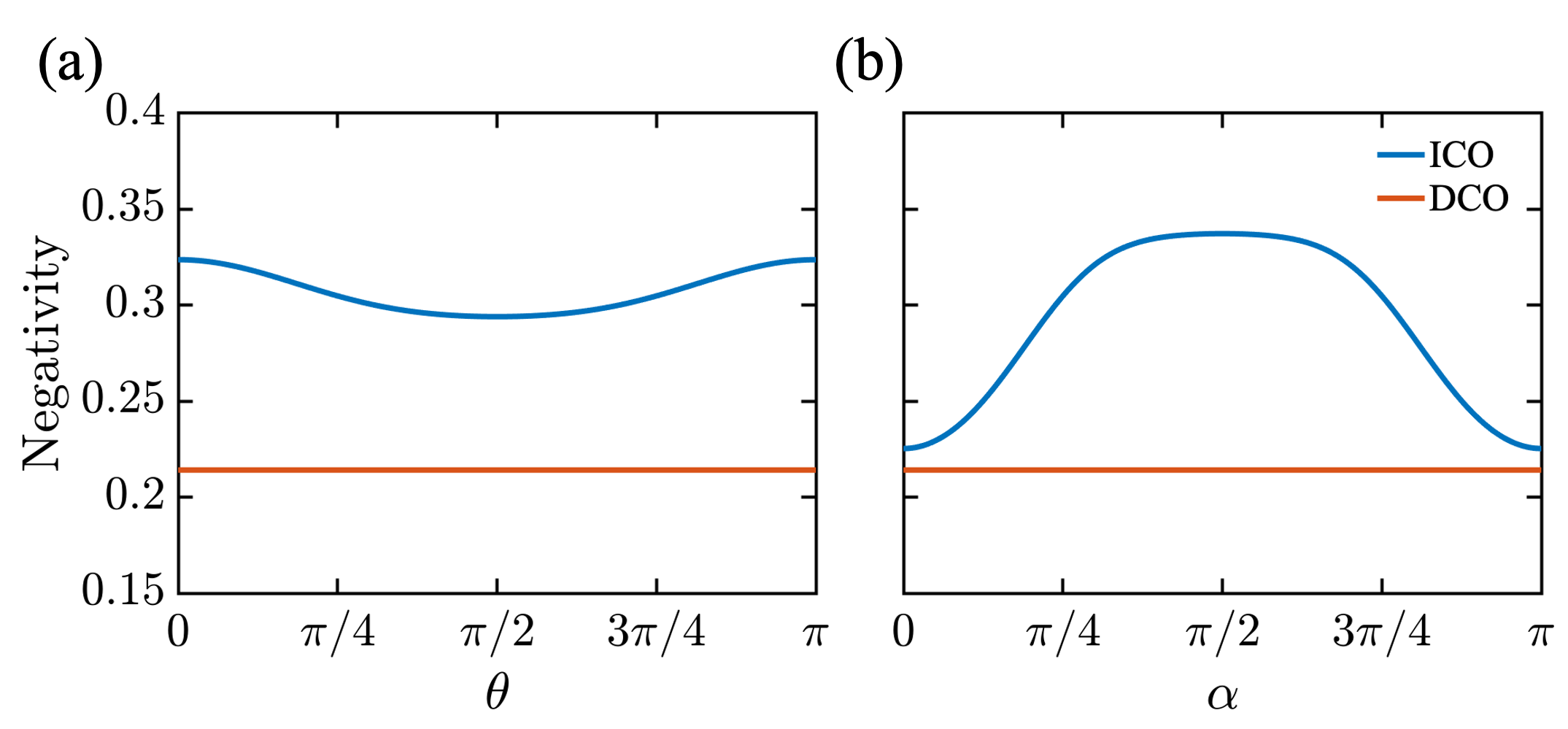}
\caption{Dependence of the negativity of the conditional state on local unitary operations under two depolarizing channels. We take equal identity Kraus weights $p_1=p_2=0.8$, corresponding to $\lambda_1=\lambda_2\approx 0.787$, and $\alpha^{(A)}=\alpha^{(B)}=\alpha$, $\theta^{(A)}=\theta^{(B)}=\theta$, $\phi^{(A)}=\phi^{(B)}=\phi_0$. Panels (a)--(b) show the dependence of negativity on $\theta$ and $\alpha$, respectively, with the remaining two parameters fixed at $\pi/4$. Orange and blue curves denote the DCO output and the $(+)$ ICO conditional output, respectively. The $(+)$ outcome attains $\mathcal N_{\mathrm{ICO}}^{U}$ and has larger negativity than the DCO output throughout the plotted parameter ranges.}
\label{fig:general-local-unitary}
\end{figure}
The flat DCO curves in \cref{fig:general-local-unitary} are consistent with this covariance argument. For the postselected quantum-switch output, the inserted local unitary changes the conditional-state negativity; throughout the plotted scans, this negativity remains above the definite-order value.

\subsubsection{Correlation between path difference and conditional-state negativity gain}
\label{non_commu}

The $X$-basis conditional outputs make explicit how operator noncommutativity enters through the path-difference operator. As shown in \cref{eq:conditional-Q}, the $(-)$ outcome extracts the component $\mathcal Q_\rho(U)/4$, while the $(+)$ outcome contains the complementary part of the equally weighted mixture, $\rho_f^{+,0}=\rho_f^{\rm MIX}-\mathcal Q_\rho(U)/4$. The quantity $\mathcal I(U)$ therefore quantifies the overall magnitude of the component redistributed between the two outcomes by causal-order interference.

Using the parameters of \cref{fig:general-local-unitary}, we fix $p=0.8$ and two angular parameters at $\pi/4$, vary the remaining parameter, and compare the negativity gain $\Delta_\mathcal{N}^U$ with $\mathcal I(U)$ in \cref{fig:9}. Over these restricted scans, the two quantities show a clear qualitative correlation, with larger values of $\mathcal I(U)$ generally accompanying larger conditional-state negativity gains.

\begin{figure}[htbp]
\centering
\includegraphics[width=0.48\textwidth]{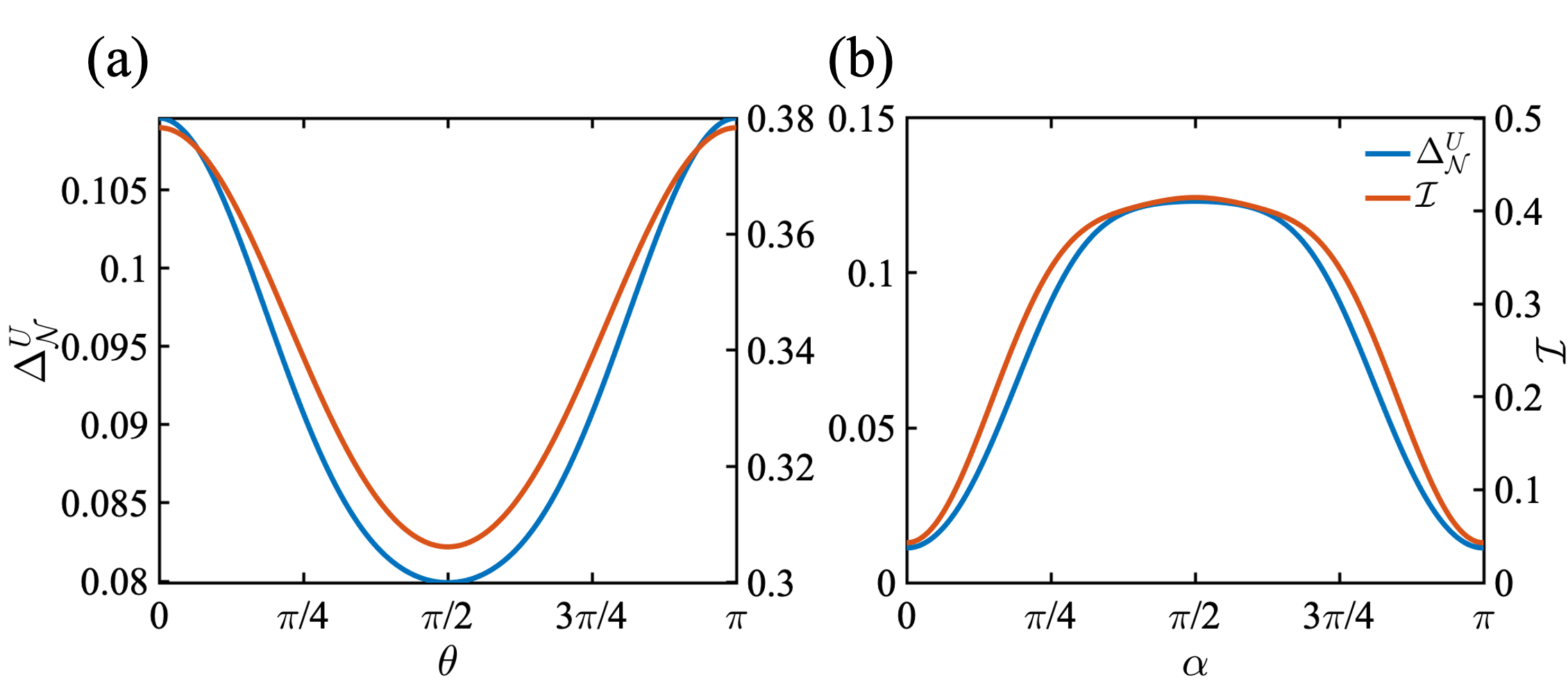}
\caption{
Conditional-state negativity gain versus path difference under local unitary operations. The blue curves denote the negativity gain $\Delta_\mathcal N^U$, and the orange curves denote $\mathcal I(U)$. We set equal identity Kraus weights $p_1=p_2=p=0.8$, corresponding to $\lambda_1=\lambda_2\approx 0.787$, and $\alpha^{(A)}=\alpha^{(B)}=\alpha$, $\theta^{(A)}=\theta^{(B)}=\theta$, $\phi^{(A)}=\phi^{(B)}=\phi_0$. Panels (a)--(b) show the dependence on $\theta$ and $\alpha$, respectively, with the remaining two angular parameters fixed at $\pi/4$.
}
\label{fig:9}
\end{figure}

\cref{spec} explains this correlation through the partial-transpose spectrum. There, $\mathcal I(U)$ bounds the spectral displacement available from $\mathcal Q_\rho(U)$, whereas an NPT crossing depends on the projection of $\mathcal Q_\rho(U)^{T_B}$ onto the relevant detection directions and on the background spectrum of the mixed output. These spectral and directional constraints account for the qualitative correlation observed in the present depolarizing family and local-unitary scans.

\section{Entanglement certification under non-Pauli noise}

For non-Pauli noise, the relation $H_{ij}^{\rightarrow}=\pm H_{ij}^{\leftarrow}$ generally no longer holds, and the control measurement no longer performs an exact binary path selection. The two ordered products of each path pair still add coherently in the conditional Kraus operators of \cref{eq:M_ij}, and the cross-order terms in \cref{eq:final_state} still redistribute the noise contributions between the two conditional outputs. The comparison between the ICO conditional states and arbitrary classical mixtures of the two definite orders therefore proceeds exactly as before.

We illustrate this case with two amplitude damping channels (ADCs), each acting locally and independently on $A$ and $B$ with noise parameter $\gamma_k$ for $k=1,2$. The input is $\rho_{AB}=|\Phi^+\rangle\langle\Phi^+|$, the inserted local operation is fixed to the identity $U=I_A\otimes I_B$, and the control qubit is measured in the $X$ basis. Each channel is a tensor product of single-qubit amplitude damping Kraus operators,
\begin{equation}
\mathcal E_{\mathrm{ADC}}^{(2q)}(\rho)
= \sum_{a,b\in \{0,1\}}
(E_a\otimes E_b)\rho(E_a\otimes E_b)^\dagger,
\end{equation}
\begin{equation}
E_0=
\begin{pmatrix}
1&0\\
0&\sqrt{1-\gamma}
\end{pmatrix},
\quad E_1=
\begin{pmatrix}
0&\sqrt{\gamma}\\
0&0 
\end{pmatrix},
\quad \gamma\in [0,1].
\end{equation}
With the identity insertion, the composition of the two ADCs is again an amplitude damping channel with effective parameter $\gamma_{\rm eff}=1-(1-\gamma_1)(1-\gamma_2)$ regardless of the order, so the two definite-order outputs and all their classical mixtures coincide. Direct multiplication further shows that the two ordered products of each single-qubit Kraus pair are proportional, for example $E_0^{(2)}E_1^{(1)}=\sqrt{\gamma_1} |0\rangle\langle 1|$ while $E_1^{(1)}E_0^{(2)}=\sqrt{1-\gamma_2} E_0^{(2)}E_1^{(1)}$, so every nonvanishing conditional Kraus operator remains a product operator with respect to $A|B$. The conditional maps therefore admit separable Kraus decompositions and Lemma~\ref{Lemma1} applies to this ADC setting.

\cref{fig:ADCs} shows the negativity gain $\Delta_{\mathcal N}^{I}(\gamma_1,\gamma_2)=\mathcal N(\widehat\rho_f^{+,0})-\mathcal N_{\mathrm{cl}}^{I}$ under the two ADCs, where the negativity of the conditional state associated with the $(+)$ outcome is the larger of the two at every sampled point. The gain is positive at every sampled interior point of the parameter grid and vanishes on the boundary, where one of the channels becomes noiseless or completely damping. However, the resulting gain is relatively small, and the NPT-certifiable regions of the ICO and DCO protocols are largely comparable. This example shows that exact Pauli sign selection is not necessary for a postselected negativity gain.
\begin{figure}[htbp]
\centering
\includegraphics[width=0.33\textwidth]{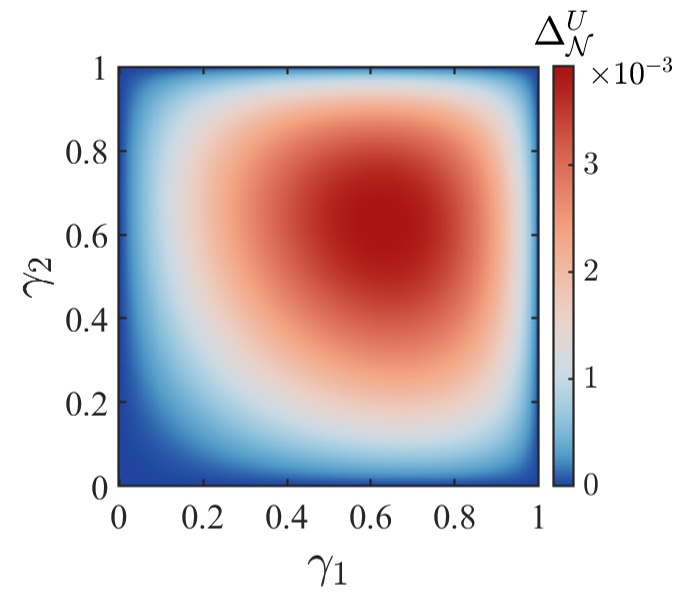}
\caption{Conditional-state negativity gain $\Delta_{\mathcal N}^{I}(\gamma_1,\gamma_2)=\mathcal N(\widehat\rho_f^{+,0})-\mathcal N_{\mathrm{cl}}^{I}$ for two successive amplitude damping channels acting locally on $A$ and $B$, with noise parameters $\gamma_1$ and $\gamma_2$ and $U=I$.}
\label{fig:ADCs}
\end{figure}

\section{High-dimensional entanglement certification}\label{sec5}

We next consider two-qutrit systems, first using negativity to certify NPT entanglement and then using nondecomposable witnesses to detect PPT entanglement.

\subsection{High-dimensional NPT certification under Weyl noise}
\label{subsec5_A}

For the qudit random unitary channels considered here, tensor-product Pauli operators are replaced by tensor-product Weyl operators. When the channel Kraus operators and the inserted unitary are Weyl operators, the two ordered products of each path pair differ by a phase rather than only by a sign. The path weights therefore become phase dependent, as stated in Corollary~\ref{cor:weyl-phase}.

\begin{corollary}
\label{cor:weyl-phase}
Suppose that the two causal orders of a path pair differ only by a phase, $H_{ij}^{\rightarrow}=e^{i\vartheta_{ij}}H_{ij}^{\leftarrow}$, as occurs when the channel Kraus operators and the inserted operation are Weyl operators. Then
\begin{equation}
M_{ij}^{\pm,\phi}
= \frac{1}{2}
\left( e^{i\vartheta_{ij}}\pm e^{-i\phi} \right)
H_{ij}^{\leftarrow}.
\end{equation}
The corresponding path weight is proportional to $\{1\pm \cos(\vartheta_{ij}+\phi)\}/2$. Hence, the $(+)$ outcome is constructive for this path when $\phi=-\vartheta_{ij}$, whereas the $(-)$ outcome is constructive when $\phi=\pi-\vartheta_{ij}$.
\end{corollary}

Consequently, every nonvanishing conditional Kraus operator remains proportional to a product Weyl operator. The conditional maps therefore admit separable Kraus decompositions, and Lemma~\ref{Lemma1} applies as in the Pauli case. Because the phases $\vartheta_{ij}$ can take more than two values, however, a qubit control cannot in general sort all Weyl paths into distinct outcomes; complete phase sorting would require a higher-dimensional or multipath control system.

Using the PPT criterion and entanglement negativity, we quantify the resulting NPT-certification performance. We consider a maximally entangled two-qutrit state $|\psi\rangle=\frac{1}{\sqrt{3}}(|00\rangle+|11\rangle+|22\rangle)$ subject to two depolarizing channels acting on the joint two-qutrit system. Their Kraus operators are built from product Weyl operators,
\begin{equation}
K_{mn,m'n'}=
\begin{cases}
\sqrt p I_9,& m=n=m'=n'=0,\\
\sqrt{\frac{1-p}{d^4-1}} D_{mn}\otimes D_{m'n'},& \text{otherwise}.
\end{cases}
\label{eq:joint-qutrit-weyl-channel}
\end{equation}
Here $d=3$, $p=[(d^4-1)\lambda+1]/d^4$, $\lambda\in [0,1]$, and $D_{mn}=X^mZ^n$ with $m,n\in \{0,1,2\}$, where
\begin{equation}
X|k\rangle=|(k+1)\operatorname{mod}3\rangle,\qquad
Z|k\rangle=\omega^k|k\rangle,\qquad
\omega=e^{i2\pi/3}.
\end{equation}
\cref{fig:3D} shows the phase diagram of $\Delta_{\mathcal N}^{U}$ and the NPT-certifiable region for $U=D_{20}\otimes D_{20}$. Throughout the scanned noise-parameter space, the fixed $(+)$ branch has negativity no smaller than that of either definite-order output. As in the two-qubit example with $U=\sigma_z\otimes \sigma_z$, the quantum-switch branch remains NPT over a broader parameter region than the definite-order outputs.
\begin{figure}[htbp]
\centering
\includegraphics[width=0.33\textwidth]{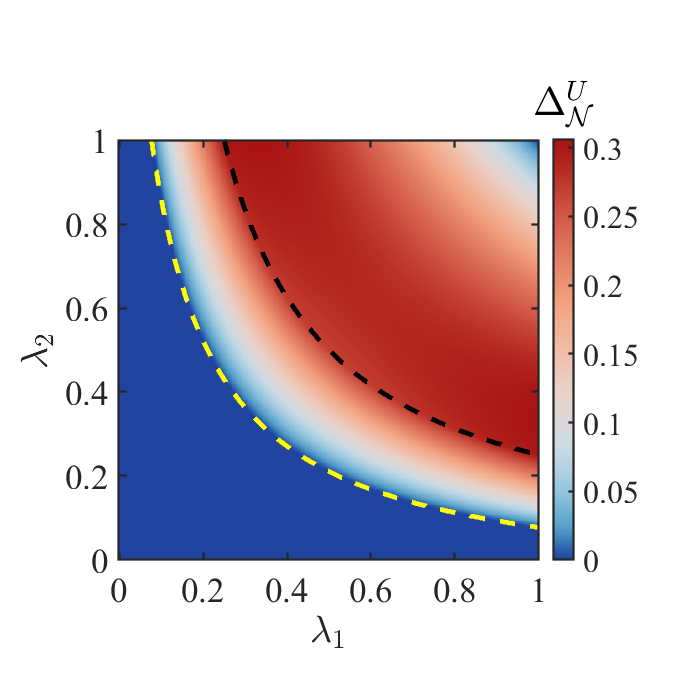}
\caption{Negativity gain $\Delta_{\mathcal N}^{U}$ with $U=D_{20}\otimes D_{20}$ for a maximally entangled two-qutrit state under two depolarizing channels acting on the joint two-qutrit system, with depolarizing parameters $\lambda_1$ and $\lambda_2$. The color scale shows the negativity gain of the fixed $(+)$ outcome relative to the corresponding definite-order outputs. The yellow and black dashed curves mark the ICO and DCO NPT-certification boundaries, respectively, so the region between them is certified only by the ICO protocol.}
\label{fig:3D}
\end{figure}

In $3\times 3$ and higher-dimensional systems, however, PPT no longer implies separability, and negativity can certify only NPT entanglement. The preceding high-dimensional analysis therefore concerns the enhancement of NPT entanglement certification. We next turn to entanglement witnesses to examine PPT entanglement.

\subsection{ICO enhancement of entanglement witness detection}

For a fixed entanglement witness $W$, linearity in the mixing parameter $\xi$ of \cref{eq:mix} implies that the minimum of $\operatorname{Tr}[W\rho_f^{\rm MIX}(\xi)]$ over all classical mixtures is attained at an endpoint,
\begin{equation}
\min_{\xi\in [0,1]}\operatorname{Tr}[W\rho_f^{\rm MIX}(\xi)]
=\min\{\operatorname{Tr}(W\rho_f^{\rightarrow}),\operatorname{Tr}(W\rho_f^{\leftarrow})\}.
\label{eq:witness-classical-benchmark}
\end{equation}
By \cref{eq:witness-classical-benchmark}, the witness expectation is nonnegative for every classical mixture if and only if it is nonnegative for both definite-order outputs. Thus, none of these outputs is detected by this witness. For the ICO protocol, the witness expectation of the unnormalized conditional output is
\begin{equation}
\operatorname{Tr}
[ W\rho_f^{\pm,\phi} ]
= \frac{1}{4}
\sum_{i,j} \operatorname{Tr}
[ W(H_{ij}^{\rightarrow}\pm e^{-i\phi}H_{ij}^{\leftarrow})
\rho_{AB}
(H_{ij}^{\rightarrow}\pm e^{-i\phi}H_{ij}^{\leftarrow})^\dagger ].
\end{equation}
Define the diagonal witness contribution as
\begin{equation}
D_W(\rho_{AB} )
= \frac{1}{4}
\sum_{i,j} ( \operatorname{Tr}[WH_{ij}^{\rightarrow}\rho_{AB}
(H_{ij}^{\rightarrow})^\dagger]
+\operatorname{Tr}[WH_{ij}^{\leftarrow}\rho_{AB}
(H_{ij}^{\leftarrow})^\dagger] ),
\end{equation}
and the witness-weighted cross-order interference term as
\begin{equation}
C_W(\rho_{AB}
,U)
= \sum_{i,j} \operatorname{Tr}
[ WH_{ij}^{\leftarrow}\rho_{AB}
(H_{ij}^{\rightarrow})^\dagger ].
\end{equation}
It follows that
\begin{equation}
\operatorname{Tr}[W\rho_f^{\pm,\phi}]
= D_W(\rho_{AB} )
\pm \frac{1}{2}
\operatorname{Re}
[ e^{-i\phi}C_W(\rho_{AB}
,U) ].
\end{equation}
Thus, cross-order interference enters the witness expectation through $C_W$. If the control outcome and phase can be optimized, the minimum unnormalized witness expectation for the chosen witness $W$ is
\begin{equation}
\min_{\pm,\phi}
\operatorname{Tr}[W\rho_f^{\pm,\phi}]
= D_W(\rho_{AB}
)-\frac{1}{2}|C_W(\rho_{AB}
,U)|.
\end{equation}
If the minimum in \cref{eq:witness-classical-benchmark} is nonnegative and $D_W(\rho_{AB})-|C_W(\rho_{AB},U)|/2<0$, then the same witness gives nonnegative expectations for both definite-order outputs and every classical mixture of them, but a negative expectation for one ICO conditional output. The negativity-based analysis of the preceding sections corresponds to the special choice $W_\eta=(|\eta\rangle\langle\eta|)^{T_B}$, associated with a negative-eigenvalue direction of the partial transpose. In the next subsection, we carry out this comparison using a nondecomposable witness, which is required for certifying PPT entangled states.

\subsection{PPT entanglement certification with the Tiles UPB state}
\label{subsec5_C}

We consider the $3\times 3$ Tiles UPB bound-entangled state
\begin{equation}
\rho_{\mathrm{UPB}}
= \frac{I_9-\Pi_{\mathrm{UPB}}}{4},
\end{equation}
where $\Pi_{\mathrm{UPB}}$ denotes the projector onto the subspace spanned by the Tiles UPB product vectors. In \cref{App_PPT}, we derive a rigorous lower bound
$
\min_{\|a\|=\|b\|=1}
\langle a,b|\Pi_{\mathrm{UPB}}|a,b\rangle
> \frac{1}{45}
$ and construct the following valid entanglement witness~\cite{PhysRevA.66.062305,LewensteinOptimizationOfPRA2000, 01042003, TERHAL200161}
\begin{equation}
W_{\mathrm{UPB}}=\Pi_{\mathrm{UPB}}-\epsilon I_9,
\qquad \epsilon=1/45.
\end{equation}
For the Weyl channels defined in \cref{eq:joint-qutrit-weyl-channel}, with the fixed local operation $U=D_{20}\otimes D_{20}$ and parameters $(p_1,p_2)=(0.97,0.97)$, we compute the outputs generated from the Tiles UPB input, namely the definite-order outputs $\rho_f^{\rightarrow}$ and $\rho_f^{\leftarrow}$ and the normalized ICO $(+)$-conditioned output $\widehat\rho_f^{+,0}$. Since the uniform mixture of all two-qutrit Weyl products is the completely depolarizing channel, each channel in \cref{eq:joint-qutrit-weyl-channel} acts as $\mathcal E_i(\rho)=\lambda_i\rho+(1-\lambda_i)I_9/9$, with $\lambda_1=\lambda_2=0.969625$ at the chosen parameters. Direct composition then gives $\rho_f^{\rightarrow}=\rho_f^{\leftarrow}=\lambda_1\lambda_2U\rho_{\mathrm{UPB}}U^\dagger+(1-\lambda_1\lambda_2)I_9/9$, so the two definite-order outputs coincide and every classical mixture of them equals the same state.

The candidate witnesses are drawn from the locally rotated witness family
\begin{equation}
\mathcal W_{\mathrm{loc}}
= \{(V_A\otimes V_B)W_{\mathrm{UPB}}(V_A\otimes V_B)^\dagger\},
\end{equation}
where $V_A,V_B\in U(3)$. Every member of $\mathcal W_{\mathrm{loc}}$ is a valid nondecomposable witness, as shown in \cref{sec:AnalyticWitness}. For the definite-order outputs, no numerical optimization is required, because the entire witness family obeys the following analytic bound. Writing $W_V=VW_{\mathrm{UPB}}V^\dagger$ with $V=V_A\otimes V_B$ and using $\operatorname{Tr}\Pi_{\mathrm{UPB}}=5$, the common definite-order output gives
\begin{equation}
\begin{aligned}
\operatorname{Tr}(W_V\rho_f^{\rightarrow})
&= \lambda_1\lambda_2\operatorname{Tr}(V\Pi_{\mathrm{UPB}}V^\dagger U\rho_{\mathrm{UPB}}U^\dagger)\\
&\quad +\frac{5(1-\lambda_1\lambda_2)}{9}-\frac{1}{45}\\
&\ge \frac{5(1-\lambda_1\lambda_2)}{9}-\frac{1}{45},
\end{aligned}
\label{eq:DcoFamilyBound}
\end{equation}
which evaluates to approximately $1.1015\times 10^{-2}$ at the chosen parameters. The first term vanishes exactly at $V=U$, since $\operatorname{Tr}(\Pi_{\mathrm{UPB}}\rho_{\mathrm{UPB}})=0$, so this bound is the global minimum over the family. Hence no witness in the continuous family $\mathcal W_{\mathrm{loc}}$ detects the definite-order outputs or any classical mixture of them.

For the ICO conditional output, we evaluate the family member that attains equality in \cref{eq:DcoFamilyBound}, namely $W_U=UW_{\mathrm{UPB}}U^\dagger$. This single witness gives $\operatorname{Tr}(W_U\widehat\rho_f^{+,0})\approx -5.9319\times 10^{-3}$ and $\operatorname{Tr}(W_U\rho_f^{\rightarrow})=\operatorname{Tr}(W_U\rho_f^{\leftarrow})\approx 1.1015\times 10^{-2}$. The same witness therefore certifies the ICO conditional output as entangled while giving a strictly positive value on both definite orders and on every classical mixture of them.

The success probability of the $(+)$ outcome is approximately $0.9701$. The negativity of $\widehat\rho_f^{+,0}$ remains zero and the minimum eigenvalue of its partial transpose is approximately $0.0026$, so the conditional output remains PPT and the detection is not obtained by converting the output into an NPT state. At this representative point, the ICO protocol therefore certifies a PPT-entangled output, whereas no witness in the locally rotated family detects either definite-order output or any classical mixture of them.

\section{Discussion and outlook}

In this work, we have shown that ICO preprocessing can enlarge the parameter region in which a fixed entanglement criterion succeeds after noise. For the Pauli and Weyl constructions and for the amplitude-damping example with $U=I$, the selected conditional maps are separability preserving across $A|B$, so entanglement detected after postselection certifies entanglement of the input. To quantify this advantage, we used the negativity gain of the conditional state for NPT-entangled inputs and identified the region in which entanglement is certifiable only by the ICO scheme. For Pauli channels, the inserted local Pauli unitary controls the sorting of Kraus paths between the two control outcomes, and $U=\sigma_z\otimes \sigma_z$ substantially enlarges the NPT-certifiable region. The input-dependent quantity $\mathcal I(U)$ measures the path-difference magnitude and tracks the conditional negativity gain within the scanned local-unitary family. Positive gains also occur for local amplitude-damping noise and in the two-qutrit Weyl setting. At a representative Tiles UPB point, a locally rotated nondecomposable witness detects the PPT postselected ICO output, while an analytic bound excludes the entire rotated witness family on the definite-order benchmark. For general local unitaries and general channels, separability preservation has not been established, so the corresponding scans compare output-state negativities.

Our analysis has also been restricted to discrete-variable systems and to noise acting only on the target states. Further directions include higher-dimensional controls for more complete Weyl phase sorting \cite{PhysRevLett.113.250402, 39vh-84n1}, continuous-variable entanglement certification \cite{Giacomini2016}, and noisy implementations of the quantum switch itself \cite{Aziz_2025}.

\begin{acknowledgments}
This work is supported by Quantum Science and Technology-National Science and Technology Major Project (Grants No. 2024ZD0302401 and No. 2021ZD0301500), National Natural Science Foundation of China (No. 12125402, No. 12534016, and No. 12405005), and Beijing Natural Science Foundation (Grant No. Z240007). S.L. acknowledges the China Postdoctoral Science Foundation (No. 2023M740119).
\end{acknowledgments}

\appendix
\section{Choi-state analysis of single-qubit channels}
\label{app_choi}

\subsection{NPT certification in the Choi state analysis}

We use the Choi state to formulate a channel-level benchmark for the postselected quantum-switch map. Specifically, we apply the conditional map induced by the quantum switch to one half of a Bell state and quantify the entanglement of the resulting Choi state by its negativity. We prepare $|\Phi^+\rangle_{RT} = (|00\rangle+|11\rangle)/\sqrt{2}$, where the reference qubit $R$ bypasses the quantum switch while the target qubit $T$ is processed by it. Let the target qubit pass through two noisy channels $\mathcal{E}_1$ and $\mathcal{E}_2$, with the intermediate operation $U$ acting only on $T$. For control state $|0\rangle_C$, the target undergoes $H_{ij}^{\rightarrow}=K_j^{(2)}UK_i^{(1)}$; for $|1\rangle_C$, it undergoes $H_{ij}^{\leftarrow}=K_i^{(1)}UK_j^{(2)}$. In this reference--target setting, the conditional Kraus operators are still given by $M_{ij}^{\pm,\phi} = \frac{1}{2}\left(H_{ij}^{\rightarrow}\pm e^{-i\phi}H_{ij}^{\leftarrow}\right)$, which define the trace-nonincreasing conditional map on the target qubit
\begin{equation}
\Lambda_{\pm,\phi}(\rho_T)
= \sum_{i,j} M_{ij}^{\pm,\phi}\rho_T M_{ij}^{\pm,\phi\dagger}.
\end{equation}
Using the Choi representation, the normalized output state on the reference--target system is
\begin{equation}
\rho_{RT}^{\mathrm{ICO},\pm,\phi}
= \frac{
(\operatorname{id}_R\otimes \Lambda_{\pm,\phi}^{(T)})
\left(|\Phi^+\rangle\langle\Phi^+|_{RT}\right)
}{
P_{\pm,\phi}
},
\end{equation}
where $P_{\pm,\phi} = \operatorname{Tr} \left[(\operatorname{id}_R\otimes \Lambda_{\pm,\phi}^{(T)})\left(|\Phi^+\rangle\langle\Phi^+|_{RT}\right)\right]$ is the postselection success probability for the corresponding control outcome, and $\operatorname{id}_R$ denotes the identity channel on the reference system. The corresponding definite order and classical mixture benchmarks are
\begin{equation}
\rho_{RT}^{12}
= (\operatorname{id}_R\otimes \mathcal E_2\circ \mathcal U\circ \mathcal E_1)
(|\Phi^+\rangle\langle\Phi^+|_{RT}),
\end{equation}
\begin{equation}
\rho_{RT}^{21}
= (\operatorname{id}_R\otimes \mathcal E_1\circ \mathcal U\circ \mathcal E_2)
(|\Phi^+\rangle\langle\Phi^+|_{RT}),
\end{equation}
and let
\begin{equation}
\rho_{RT}^{\rm MIX}(\xi)
=\xi\rho_{RT}^{12}+(1-\xi)\rho_{RT}^{21},
\qquad \xi\in [0,1],
\end{equation}
where $\mathcal U(\rho)=U\rho U^\dagger$. Because the set of PPT states is convex, a sufficient condition for a postselection advantage is a parameter region in which
\begin{equation}
\mathcal N(\rho_{RT}^{12})=0,\qquad
\mathcal N(\rho_{RT}^{21})=0,
\end{equation}
while, for some control outcome,
\begin{equation}
P_{\pm,\phi}>0,\qquad
\mathcal N(\rho_{RT}^{\mathrm{ICO},\pm,\phi})>0.
\end{equation}
Proposition~\ref{prop:choi-advantage} gives an analytically tractable postselection advantage region in the Choi state setting.
\begin{proposition}\label{prop:choi-advantage}
Consider two standard single-qubit depolarizing channels
\begin{equation}
\mathcal D_{\lambda_k}(\tau)
= \lambda_k\tau+(1-\lambda_k)\frac{I}{2},
\qquad k=1,2,
\end{equation}
where $0\le \lambda_k\le 1$, and $\tau$ denotes an arbitrary input density operator on the target qubit. We use the Pauli realization
\begin{equation}
\mathcal D_{\lambda_k}(\tau)=a_k\tau+b_k(X\tau X+Y\tau Y+Z\tau Z),
\end{equation}
with $a_k=(1+3\lambda_k)/4$ and $b_k=(1-\lambda_k)/4$. Let the intermediate operation be $U=Z$, and postselect the $(+)$ control outcome $|+\rangle$, corresponding to $\phi=0$. If
\begin{equation}
\lambda_1\lambda_2\le \frac13,
\qquad 7\lambda_1\lambda_2+\lambda_1+\lambda_2>1
\end{equation}
hold simultaneously, then both definite-order Choi outputs and every classical mixture of them are PPT, whereas the normalized ICO $(+)$-conditioned Choi output is NPT and occurs with nonzero probability. These inequalities therefore define an analytic postselection-advantage region for the conditional map generated by the quantum switch. The proof is given in \cref{App_prop}.
\end{proposition}

The analytic region above covers a continuous range of parameters. In the symmetric case $\lambda_1=\lambda_2=\lambda$, the advantage region is approximately $0.2612<\lambda\le 0.5774$. This reference--target construction provides channel-level evidence for the causal-order interference advantage. The full bipartite state certification problem is addressed in the target state model of the main text, where we ask whether quantum-switch preprocessing enlarges the parameter region in which the PPT criterion or an entanglement witness certifies $A|B$ entanglement in a given state $\rho_{AB}$.

\subsection{Proof of Proposition~\ref{prop:choi-advantage}}
\label{App_prop}

We first consider the definite order process $\mathcal D_{\lambda_2}\circ \mathcal U_Z\circ \mathcal D_{\lambda_1}$, where $\mathcal U_Z(\tau)=Z\tau Z$. Owing to the Pauli covariance of the depolarizing channel, the composition of two depolarizing channels remains depolarizing with effective parameter $\lambda_1\lambda_2$. The intermediate $Z$ operation induces only a local unitary rotation on the corresponding Choi state and therefore does not affect its entanglement. The Choi state of a single-qubit depolarizing channel is Bell-diagonal, with maximal Bell weight $(1+3\lambda)/4$. For a two-qubit Bell-diagonal state, the negativity is given by $\max\{0,p_{\max}-1/2\}$, implying that NPT occurs if and only if $\lambda>1/3$. Consequently, both definite-order Choi states, and hence every classical mixture of them, are PPT whenever $\lambda=\lambda_1\lambda_2\le 1/3$.

For the ICO $(+)$ outcome, each pair of Pauli paths $P_i,Q_j\in \{I,X,Y,Z\}$ gives two causal orders, $H_{ij}^{\rightarrow}=\sqrt{p_iq_j}Q_jZP_i$ and $H_{ij}^{\leftarrow}=\sqrt{p_iq_j}P_iZQ_j$, with conditional Kraus operator $M_{ij}^{+}=\sqrt{p_iq_j}(Q_jZP_i+P_iZQ_j)/2$. Since Pauli operators either commute or anticommute, the $(+)$ outcome retains only the paths satisfying $Q_jZP_i=P_iZQ_j$. Grouping the surviving terms by the resulting effective Pauli operator gives the unnormalized Pauli weights
\begin{equation}
\begin{aligned}
w_Z=a_1a_2+3b_1b_2=\frac{1+3\lambda_1\lambda_2}{4},\\
w_I=a_1b_2+b_1a_2=\frac{1+\lambda_1+\lambda_2-3\lambda_1\lambda_2}{8},
\end{aligned}
\end{equation}
and $w_X=w_Y=2b_1b_2=(1-\lambda_1)(1-\lambda_2)/8$. The success probability of the $(+)$ outcome is
\begin{equation}
P_+
= w_I+w_X+w_Y+w_Z
= \frac{5+5\lambda_1\lambda_2-\lambda_1-\lambda_2}{8},
\end{equation}
which is strictly positive for $0\le \lambda_1,\lambda_2\le 1$. The maximal Bell weight of the normalized Choi state is
\begin{equation}
p_{\max}^{(+)}
= \frac{w_Z}{P_+}
= \frac{2(1+3\lambda_1\lambda_2)}
{5+5\lambda_1\lambda_2-\lambda_1-\lambda_2}.
\end{equation}
Using the Bell-diagonal two-qubit formula $\mathcal N=\max\{0,p_{\max}-1/2\}$, the negativity of the ICO $(+)$-conditioned output follows, with the NPT condition $7\lambda_1\lambda_2+\lambda_1+\lambda_2>1$.

\section{Proof of Lemma~\ref{Lemma1}}
\label{proof_L1}

We first show that entanglement cannot be generated from a separable input. For a separable state $\sigma_{AB}=\sum_\ell q_\ell \sigma_\ell^A \otimes \sigma_\ell^B$ with $q_\ell\ge 0$ and $\sum_\ell q_\ell=1$, the conditional output reads
\begin{equation}
\Lambda_r(\sigma_{AB}) = \sum_{\ell,\kappa} q_\ell (A_{r,\kappa}\sigma_\ell^A A_{r,\kappa}^\dagger) \otimes (B_{r,\kappa}\sigma_\ell^B B_{r,\kappa}^\dagger),
\end{equation}
where each term is a positive product operator. If $\operatorname{Tr}\Lambda_r(\sigma_{AB})>0$, normalization simply rescales the convex combination to sum to unity, so $\sigma_{AB}^{(r)}$ remains separable. Thus, no entanglement is generated from a separable input by any conditional map.

We now prove the average monotonicity of the probability-weighted negativity of the conditional states. Let $\Xi=\rho_{AB}^{T_B}$ and define $\Phi_r=\Gamma_B\circ \Lambda_r\circ \Gamma_B$, where $\Gamma_B$ denotes partial transposition on subsystem $B$. Since each $\Phi_r$ is positive and the total map $\sum_r\Lambda_r$ is trace preserving, the flagged map
\begin{equation}
\widetilde{\Phi}(\Xi)
= \sum_r
|r\rangle\langle r|\otimes \Phi_r(\Xi)
\end{equation}
is also positive and trace preserving. Here the orthogonal classical register records the postselection outcome $r$. Let $\Xi=\Xi_+-\Xi_-$ be the positive--negative decomposition of the Hermitian operator $\Xi$, with $\Xi_+,\Xi_-\ge 0$. Positivity and trace preservation of $\widetilde{\Phi}$ imply
\begin{equation}
\begin{aligned}
\|\widetilde{\Phi}(\Xi)\|_1
&\le \operatorname{Tr}\widetilde{\Phi}(\Xi_+)
+ \operatorname{Tr}\widetilde{\Phi}(\Xi_-)\\
&= \operatorname{Tr}\Xi_+ + \operatorname{Tr}\Xi_-\\
&= \|\Xi\|_1 .
\end{aligned}
\end{equation}
Because different $r$ blocks of the classical register are mutually orthogonal, $\|\widetilde{\Phi}(\Xi)\|_1=\sum_r\|\Phi_r(\Xi)\|_1$, and hence $\sum_r\|\Phi_r(\Xi)\|_1\le \|\Xi\|_1$. Moreover, $\Phi_r(\Xi)=[\Lambda_r(\rho_{AB})]^{T_B}$ and $p_r=\operatorname{Tr}\Lambda_r(\rho_{AB})=\operatorname{Tr}\Phi_r(\Xi)$. Therefore, the probability-weighted negativity of each conditional state reads
\begin{equation}
p_r\mathcal N(\rho_{AB}^{(r)})
= \frac{\|\Phi_r(\Xi)\|_1-p_r}{2}.
\end{equation}
It follows that
\begin{equation}
\sum_r p_r\mathcal N(\rho_{AB}^{(r)})
= \frac{1}{2}
\left( \sum_r\|\Phi_r(\Xi)\|_1 - \sum_r p_r \right).
\end{equation}
Since $\sum_r\Lambda_r$ is trace preserving, $\sum_rp_r=1$. Thus,
\begin{equation}
\sum_r p_r\mathcal N(\rho_{AB}^{(r)})
\le \frac{\|\Xi\|_1-1}{2}
= \mathcal N(\rho_{AB}).
\end{equation}
Finally, since each term $p_r\mathcal N(\rho_{AB}^{(r)})$ is non-negative, each individual outcome also satisfies $p_r\mathcal N(\rho_{AB}^{(r)})\le \mathcal N(\rho_{AB})$. However, once labeled by a classical outcome, the conditional Kraus operators of the quantum switch remain coherent sums of the corresponding definite-order Kraus operators rather than reducing to a classical selection of the two orders. Therefore, they do not automatically admit a separable Kraus decomposition or satisfy the PPT-preserving condition. For special cases such as Pauli or Weyl noise, the sign or phase relation between the two ordered products in Theorem~\ref{thm:pauli-parity} and Corollary~\ref{cor:weyl-phase} reduces each conditional map to a selection and reweighting of local Pauli/Weyl paths, so that the corresponding locality or PPT-preserving condition can be satisfied.

\section{Input-state dependence of the ICO advantage in entanglement certification}
\label{App_robust}

The negativity analysis in the main text uses maximally entangled target states. To examine the dependence on the Schmidt coefficients, we restrict to the computational-basis family
\begin{equation}
|\psi_\varphi\rangle=\cos{\varphi}|00\rangle+\sin{\varphi}|11\rangle,
\qquad \varphi\in [0,\pi/4].
\label{eq:pure_state}
\end{equation}
Every pure two-qubit state is local-unitary equivalent to a state of this form. Because the inserted unitary $U=\sigma_z\otimes \sigma_z$ is fixed in the computational basis, different local basis choices are not automatically equivalent and would require rotating $U$ accordingly; we therefore scan $\varphi$ within the family above. Both channels are global depolarizing channels on $AB$ with equal identity Kraus weights $p_1=p_2=p$. \cref{fig:state-change-pauliz} compares the definite-order and fixed $(+)$-conditioned quantum-switch outputs as functions of $\varphi$ for several values of $p$.
\begin{figure}[htpb]
\centering
\includegraphics[width=0.35\textwidth]{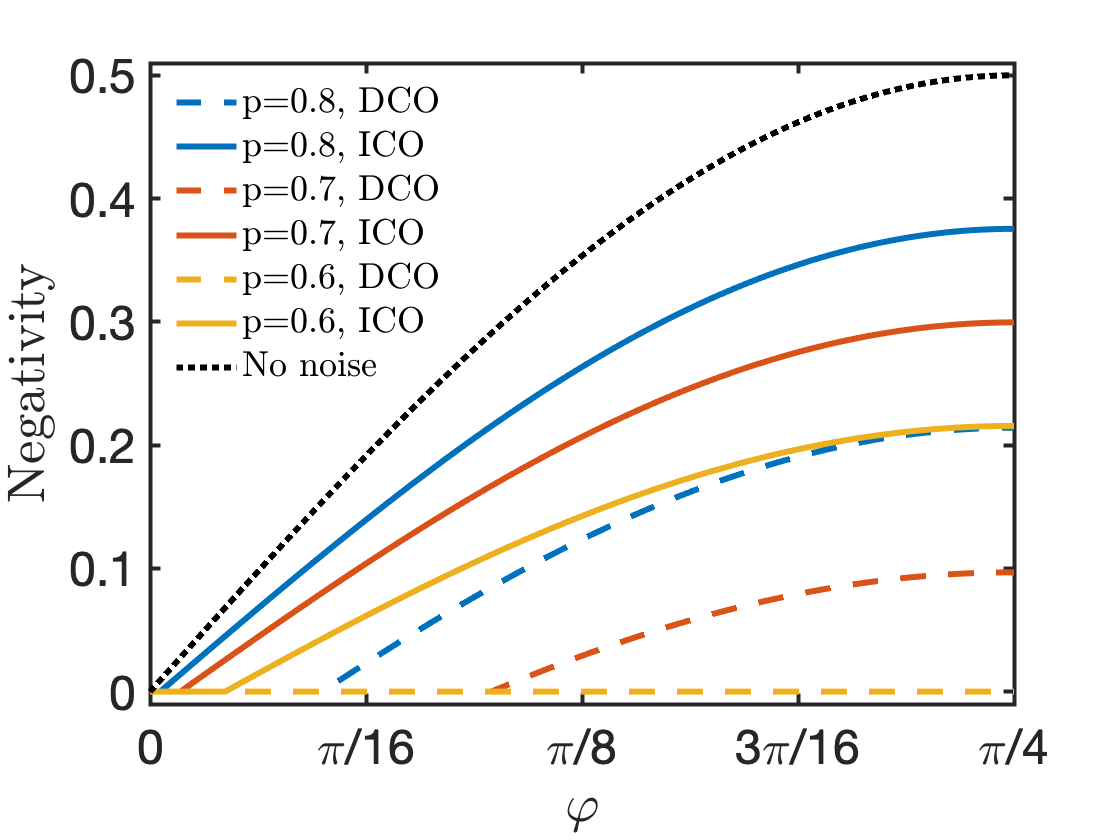}
\caption{Negativity under two depolarizing channels with the inserted local unitary $U=\sigma_z\otimes \sigma_z$. The negativity is plotted versus the Schmidt parameter $\varphi$ of the input state. Dashed and solid curves correspond to the DCO and postselected ICO outputs, respectively, for different depolarizing parameters $p$, while the black dotted curve denotes the noiseless reference. The postselected ICO output has a negativity no smaller than the corresponding DCO output throughout the plotted range.
}
\label{fig:state-change-pauliz}
\end{figure}

For every plotted value of $\varphi$, the fixed $(+)$-conditioned output has negativity no smaller than that of the definite-order output at the same $p$. Across the plotted noise strengths, the conditional output also remains NPT in cases where the definite-order output has already become PPT.

\section{Spectral analysis of path difference and the negativity of conditional states}
\label{spec}

To clarify how $\mathcal{Q}_\rho(U)$ drives the partial-transpose spectrum of the conditional outputs across the zero-eigenvalue boundary, we provide an explicit spectral representation of the negativity of the conditional states. Let $X_\pm =\rho_f^{\pm,0}$ denote the unnormalized conditional output states for $\phi=0$. From the expansion of the unnormalized conditional states, we obtain
\begin{equation}
X_- = \frac{1}{4}\mathcal{Q}_\rho(U),
\qquad X_+ = \rho_f^{\mathrm{MIX}} - \frac{1}{4}\mathcal{Q}_\rho(U),
\label{eq:Qfilter}
\end{equation}
with corresponding success probabilities $P_- = \mathrm{Tr}[\mathcal{Q}_\rho(U)]/4$ and $P_+ = 1 - \mathrm{Tr}[\mathcal{Q}_\rho(U)]/4$. The normalized conditional states are denoted by $\rho_\pm =\widehat{\rho}_f^{\pm,0}$ and take the form
\begin{equation}
\rho_- = \frac{\mathcal{Q}_\rho(U)}{\mathrm{Tr}[\mathcal{Q}_\rho(U)]},
\qquad \rho_+ =
\frac{\rho_f^{\mathrm{MIX}}-\mathcal{Q}_\rho(U)/4}
{1-\mathrm{Tr}[\mathcal{Q}_\rho(U)]/4}.
\end{equation}

\begin{proposition}[Variational form of conditional-state negativity and the PPT criterion]
\label{prop:NPT-filter}

Let $P_+>0$. The negativity of the $(+)$-conditioned state under partial transpose on subsystem $B$ reads
\begin{equation}
\begin{aligned}
\mathcal N(\rho_+)
&= \frac{1}{P_+}\operatorname{Tr} \left[\left(\frac{1}{4}\mathcal Q_\rho(U)^{T_B}
- (\rho_f^{\mathrm{MIX}})^{T_B}\right)_+\right] \\
&= \frac{1}{P_+}\max_{0\le \Pi \le I}
\left[ \frac{1}{4}\operatorname{Tr}(\Pi \mathcal Q_\rho(U)^{T_B})
- \operatorname{Tr}(\Pi (\rho_f^{\mathrm{MIX}})^{T_B}) \right],
\end{aligned}
\end{equation}
where $(\cdot)_+$ denotes the positive spectral part. If $\operatorname{Tr}\mathcal Q_\rho(U)>0$, then
\begin{equation}
\rho_-^{T_B}=\frac{\mathcal Q_\rho(U)^{T_B}}{\operatorname{Tr}\mathcal Q_\rho(U)},
\qquad \mathcal N(\rho_-)=
\frac{\operatorname{Tr}[(-\mathcal Q_\rho(U)^{T_B})_+]}{\operatorname{Tr}\mathcal Q_\rho(U)}.
\label{eq:rho-minus-negativity}
\end{equation}
Thus $\rho_+$ is NPT iff there exists a normalized $|\eta\rangle$ such that
\begin{equation}
\frac{1}{4}\langle\eta|\mathcal Q_\rho(U)^{T_B}|\eta\rangle
> \langle\eta|(\rho_f^{\mathrm{MIX}})^{T_B}|\eta\rangle.
\label{eq:spec_dis}
\end{equation}
\end{proposition}

The result follows from $\rho_+^{T_B}=((\rho_f^{\mathrm{MIX}})^{T_B}-\mathcal Q_\rho(U)^{T_B}/4)/P_+$, the definition of negativity, and $\operatorname{Tr}(C_+)=\max_{0\le \Pi\le I}\operatorname{Tr}(\Pi C)$, with optimal $\Pi$ given by the positive spectral projector of $C=\mathcal Q_\rho(U)^{T_B}/4-(\rho_f^{\mathrm{MIX}})^{T_B}$.

The same criterion identifies the eigendirections that are most susceptible to crossing the PPT boundary. Let $(\rho_f^{\mathrm{MIX}})^{T_B}|m_\ell\rangle=\nu_\ell|m_\ell\rangle$. When the mixed output is PPT or close to the PPT boundary, eigenvectors with small $\nu_\ell$ identify the most sensitive directions for NPT threshold crossing. Along such directions,
\begin{equation}
\langle m_\ell|\rho_+^{T_B}|m_\ell\rangle
=\frac{1}{P_+}\left(\nu_\ell-\frac{1}{4}\langle m_\ell|\mathcal Q_\rho(U)^{T_B}|m_\ell\rangle\right),
\end{equation}
so that the condition $\langle m_\ell|\mathcal Q_\rho(U)^{T_B}|m_\ell\rangle>4\nu_\ell$ yields an NPT certification.

If $(\rho_f^{\mathrm{MIX}})^{T_B}$ and $\mathcal Q_\rho(U)^{T_B}$ are approximately jointly diagonalizable on the relevant subspace, then
\begin{equation}
\mathcal N(\rho_+)
\approx \frac{1}{P_+}\sum_\ell
\left[ \frac{1}{4}\langle m_\ell|\mathcal Q_\rho(U)^{T_B}|m_\ell\rangle
-\nu_\ell
\right]_+,
\label{eq:neg-diagonal}
\end{equation}
with equality when $[(\rho_f^{\mathrm{MIX}})^{T_B},\mathcal Q_\rho(U)^{T_B}]=0$. In this sense, NPT in the $(+)$-conditioned state corresponds to directional crossings of the zero eigenvalue boundary, determined by the competition between the two spectra along $|m_\ell\rangle$. Equivalently, \cref{eq:spec_dis} is the witness condition $\operatorname{Tr}(W_\eta\rho_+)<0$ for $W_\eta=(|\eta\rangle\langle\eta|)^{T_B}$, and replacing $W_\eta$ by a general entanglement witness gives the corresponding witness criterion.

For Pauli channels and Pauli unitaries, this spectral structure follows directly from Theorem~\ref{thm:pauli-parity}. The $(-)$ outcome retains exactly the paths with $s_{ij}(U)=1$, so that $X_-=\mathcal Q_\rho(U)/4$ collects their contribution, whereas $X_+=\rho_f^{\mathrm{MIX}}-\mathcal Q_\rho(U)/4$ collects the complementary paths retained by the $(+)$ outcome. By \cref{eq:spec_dis}, the normalized $(+)$-conditioned state is NPT whenever the expectation value of $\mathcal Q_\rho(U)^{T_B}/4$ exceeds that of $(\rho_f^{\mathrm{MIX}})^{T_B}$ for some normalized $|\eta\rangle$. Under $X$-basis control measurements, $\mathcal N_{\mathrm{ICO}}^{U}$ is given by the larger of the negativities of the two conditional output states. The operator $\mathcal Q_\rho(U)$ therefore connects the negativity gain $\Delta_\mathcal N^U$ to the operator noncommutativity along the two ordered paths, through the projection of $\mathcal Q_\rho(U)^{T_B}$ onto the eigenvectors of $(\rho_f^{\mathrm{MIX}})^{T_B}$ with near-zero eigenvalue.

Further, within the above spectral representation, $\mathcal I(U)=\|\mathcal Q_\rho(U)\|_{\mathrm{HS}}$ quantifies the overall magnitude of the path-difference contribution. By the Cauchy--Schwarz inequality for the Hilbert--Schmidt inner product and $\|W_\eta\|_{\mathrm{HS}}=1$, one has $|\operatorname{Tr}[W_\eta\mathcal Q_\rho(U)]|\le \|\mathcal Q_\rho(U)\|_{\mathrm{HS}}=\mathcal I(U)$ for any normalized witness direction $W_\eta=(|\eta\rangle\langle\eta|)^{T_B}$. Hence the displacement term $\operatorname{Tr}[W_\eta\mathcal Q_\rho(U)]/4$ on the left-hand side of \cref{eq:spec_dis} is bounded in magnitude by $\mathcal I(U)/4$, while the actual enhancement is determined by the projection $\operatorname{Tr}[W_\eta\mathcal Q_\rho(U)]$ along the relevant NPT detection directions.

When the mixed output is PPT, so that $(\rho_f^{\mathrm{MIX}})^{T_B}\ge 0$, let $\nu_{\min}$ denote the smallest eigenvalue of $(\rho_f^{\mathrm{MIX}})^{T_B}$. Using invariance of the Hilbert--Schmidt norm under partial transposition and $\|A\|_\infty \le \|A\|_{\mathrm{HS}}$, one obtains $\mathcal Q_\rho(U)^{T_B}\le \mathcal I(U)  I$ together with $(\rho_f^{\mathrm{MIX}})^{T_B}\ge \nu_{\min} I$, which yields the necessary threshold condition
\begin{equation}
\mathcal N(\rho_+)>0
\ \Rightarrow \
\mathcal I(U)>4\nu_{\min}.
\label{eq:I-threshold}
\end{equation}
Thus, NPT in the $(+)$-conditioned state can only arise once the path-difference magnitude $\mathcal I(U)$ exceeds the spectral gap $4\nu_{\min}$ separating the classical mixed output from the PPT boundary.

For the restricted random-unitary family considered here, crossing the PPT boundary is controlled by the projections of $\mathcal Q_\rho(U)^{T_B}$ onto the near-zero eigendirections of $(\rho_f^{\mathrm{MIX}})^{T_B}$, each of which is bounded in magnitude by $\mathcal I(U)$. A larger $\mathcal I(U)$ therefore permits, but does not guarantee, a larger negativity gain. The close covariation in \cref{fig:9} is consistent with the relevant projections remaining substantial across the scanned local-unitary family.

\section{Construction of the Tiles UPB witness}
\label{App_PPT}

This appendix proves that $W_{\mathrm{UPB}}=\Pi_{\mathrm{UPB}}-I_9/45$ is a valid nondecomposable witness for the Tiles UPB state. \cref{sec:AnalyticWitness} reduces validity to the bound $\mu_*>1/45$, and \cref{subsec:GramIdentity} proves the polynomial inequality used to establish that bound.

\subsection{Validity of the Tiles UPB witness}
\label{sec:AnalyticWitness}

We use the Tiles construction introduced by Bennett \textit{et al.}~\cite{BennettUnextendibleProductPRL1999}. In the standard basis $\{\ket0,\ket1,\ket2\}$ of each qutrit, define the following two-qutrit product states:
\begin{equation}
\begin{aligned}
\ket{\phi_1}&=\ket0\otimes \frac{\ket0-\ket1}{\sqrt2},
\quad \ket{\phi_2}=\frac{\ket0-\ket1}{\sqrt2}\otimes \ket2,\\
\ket{\phi_3}&=\ket2\otimes \frac{\ket1-\ket2}{\sqrt2},\quad
\ket{\phi_4}=\frac{\ket1-\ket2}{\sqrt2}\otimes \ket0,\\
\ket{\phi_5}&=\frac{\ket0+\ket1+\ket2}{\sqrt3}\otimes \frac{\ket0+\ket1+\ket2}{\sqrt3}.
\end{aligned}
\label{eq:TilesVectors}
\end{equation}
These vectors are pairwise orthogonal and normalized. Let $\Pi_{\mathrm{UPB}}=\sum_{i=1}^{5}|\phi_i\rangle\langle\phi_i|$ and $\rho_{\mathrm{UPB}}=(I_9-\Pi_{\mathrm{UPB}})/4$. The witness threshold is the minimum overlap of $\Pi_{\UPB}$ with a normalized product state,
\begin{equation}
\mu_*
=\min_{\norm a=\norm b=1}
\bra{a,b}\Pi_{\UPB}\ket{a,b}.
\label{eq:MuDefinition}
\end{equation}
Compactness of the set of product states and the unextendibility of the UPB imply that $\mu_*>0$~\cite{TERHAL200161}. Determining a certified threshold $\epsilon\le \mu_*$ is the key step in constructing a witness of this form. For any $\epsilon>0$, define
\begin{equation}
W_\epsilon=\Pi_{\UPB}-\epsilon\Id_9.
\label{eq:ShiftedWitness}
\end{equation}
This operator has nonnegative expectation values on all separable states if and only if $\epsilon\leq \mu_*$, while $\Tr(W_\epsilon\rho_{\UPB})=-\epsilon<0$. Thus, for $0<\epsilon\leq \mu_*$, $W_\epsilon$ detects the PPT state $\rho_{\UPB}$.

Write $\ket{\phi_i}=\ket{\alpha_i}\ket{\beta_i}$. For fixed $b\in \mathbb C^3$, let $F(b)=\sum_i|\langle b|\beta_i\rangle|^2\dyad{\alpha_i}{\alpha_i}$, so that $\bra{a,b}\Pi_{\UPB}\ket{a,b}=\bra aF(b)\ket a$. The Rayleigh--Ritz principle gives
\begin{equation}
\mu_*=\min_{\norm b=1}\lambda_{\min}[F(b)].
\label{eq:MuReduced}
\end{equation}
For a real vector $b=(x,y,z)^{\mathsf T}$, direct substitution yields
\begin{equation}
\begin{aligned}
F(b)={}&
\frac{(x-y)^2}{2}\dyad{0}{0}
+\frac{z^2}{2}(\ket0-\ket1)(\bra0-\bra1)\\
&+\frac{(y-z)^2}{2}\dyad{2}{2}
+\frac{x^2}{2}(\ket1-\ket2)(\bra1-\bra2)\\
&+\frac{(x+y+z)^2}{9}\dyad{\bm e}{\bm e},
\end{aligned}
\label{eq:ReducedMatrixExplicit}
\end{equation}
where $\ket{\bm e}=\ket0+\ket1+\ket2$ is unnormalized. The matrix $F(b)$ is positive semidefinite and is a homogeneous quadratic function of $x,y,z$.

For an arbitrary $3\times 3$ matrix, let $e_2(F)=[\Tr(F)^2-\Tr(F^2)]/2$. For nonzero $b$, let the eigenvalues of $F(b)$ be $0\leq \zeta_1\leq \zeta_2\leq \zeta_3$. The Gram-matrix certificate proved in \cref{subsec:GramIdentity} yields
\begin{equation}
45\det F(b)-\norm b^2e_2[F(b)]>0.
\label{eq:DeterminantInequality}
\end{equation}
Since $F(b)\geq 0$ and \cref{eq:DeterminantInequality} holds, $\det F(b)>0$, and hence $F(b)>0$ and $e_2[F(b)]>0$. Moreover, because $e_2(F)\geq \zeta_2\zeta_3$,
\begin{equation}
\lambda_{\min}[F(b)]
=\frac{\det F(b)}{\zeta_2\zeta_3}
\geq \frac{\det F(b)}{e_2[F(b)]}
>\frac{\norm b^2}{45}.
\label{eq:RealEigenvalueBound}
\end{equation}

For a complex vector $b=u+\mathrm{i}v$, where $u,v\in \mathbb R^3$, all $\beta_i$ are real vectors, and therefore $|\braket{\beta_i}{b}|^2=\braket{\beta_i}{u}^2+\braket{\beta_i}{v}^2$. It follows that $F(b)=F(u)+F(v)$. Using the Hermitian-matrix inequality $\lambda_{\min}(L_1+L_2)\geq \lambda_{\min}(L_1)+\lambda_{\min}(L_2)$ and applying \cref{eq:RealEigenvalueBound} separately to the nonzero components of $u$ and $v$, we obtain, for every $b\neq 0$,
\begin{equation}
\lambda_{\min}[F(b)]
>\frac{\norm u^2+\norm v^2}{45}
=\frac{\norm b^2}{45}.
\label{eq:ComplexEigenvalueBound}
\end{equation}
The function $\lambda_{\min}[F(b)]$ is continuous on the complex unit sphere and therefore attains its minimum. Together with \cref{eq:ComplexEigenvalueBound}, this shows that $\mu_*>1/45$.

We therefore fix $\epsilon=1/45$ and take $W_{\mathrm{UPB}}=\Pi_{\mathrm{UPB}}-I_9/45$ as the witness used in the main text. The state $\rho_{\mathrm{UPB}}$ is PPT and the minimum eigenvalue of its partial transpose equals zero, so $W_{\mathrm{UPB}}$ detects a PPT state and is therefore nondecomposable \cite{LewensteinOptimizationOfPRA2000}. Moreover, for any product unitary $V_A\otimes V_B$, the vector $(V_A\otimes V_B)^\dagger|a,b\rangle$ remains a normalized product state, so $(V_A\otimes V_B)W_{\mathrm{UPB}}(V_A\otimes V_B)^\dagger$ is again a valid witness, and it remains nondecomposable because it detects the PPT state $(V_A\otimes V_B)\rho_{\mathrm{UPB}}(V_A\otimes V_B)^\dagger$. Local rotations of $W_{\mathrm{UPB}}$ thus generate the witness family compared on the ICO and DCO outputs in \cref{subsec5_C}.

\subsection{A Gram identity for a sextic polynomial}
\label{subsec:GramIdentity}

To prove \cref{eq:DeterminantInequality}, note that its left-hand side is a homogeneous polynomial of degree six. To remove the rational denominators in $F(b)$, define
\begin{equation}
S(x,y,z)
= 72\left[
45\det F(b)
- (x^2+y^2+z^2)e_2\bigl(F(b)\bigr) \right].
\label{eq:P-definition}
\end{equation}
Thus, $S(x,y,z)$ is a homogeneous sextic polynomial with integer coefficients.

There are ten homogeneous cubic monomials in three variables. Define
\[
m(x,y,z)=(x^3,x^2y,x^2z,xy^2,xyz,xz^2,y^3,y^2z,yz^2,z^3)^{\mathsf T}.
\]
For any real symmetric $10\times 10$ matrix $G$, $m^{\mathsf T}Gm$ is a homogeneous sextic polynomial. Therefore, a real symmetric matrix $G$ satisfying
\begin{equation}
S(x,y,z)
= m(x,y,z)^{\mathsf T}Gm(x,y,z)
\label{eq:gram-identity}
\end{equation}
is a Gram matrix of $S(x,y,z)$ with respect to the monomial basis $m$.

It remains to exhibit a positive-definite matrix satisfying \cref{eq:gram-identity}. The feasible Gram matrices form an affine space. We first located a positive-definite candidate by semidefinite programming. Because a floating-point solution is not an analytic certificate of positivity, we then reconstructed the candidate in exact arithmetic, obtaining the integer Gram matrix
\begin{widetext}
\begin{equation}
G=
\begin{pmatrix}
292 & 12 & 328 & -283 & -326 & -9 & -41 & 116 & -15 & -27\\
12 & 275 & -591 & -252 & 17 & -295 & -2 & -43 & 83 & -15\\
328 & -591 & 2028 & 265 & -280 & 1043 & -121 & 386 & -295 & -9\\
-283 & -252 & 265 & 557 & 247 & 386 & 10 & -152 & -43 & 116\\
-326 & 17 & -280 & 247 & 1002 & -280 & 42 & 247 & 17 & -326\\
-9 & -295 & 1043 & 386 & -280 & 2028 & -121 & 265 & -591 & 328\\
-41 & -2 & -121 & 10 & 42 & -121 & 56 & 10 & -2 & -41\\
116 & -43 & 386 & -152 & 247 & 265 & 10 & 557 & -252 & -283\\
-15 & 83 & -295 & -43 & 17 & -591 & -2 & -252 & 275 & 12\\
-27 & -15 & -9 & 116 & -326 & 328 & -41 & -283 & 12 & 292
\end{pmatrix}.
\label{eq:GramCertificate}
\end{equation}
\end{widetext}
Expanding $m^{\mathsf T}Gm$ and matching coefficients verifies \cref{eq:gram-identity} exactly. All ten leading principal minors of $G$ are positive, therefore $G>0$. For every nonzero $(x,y,z)\in \R^3$, the monomial vector $m$ is nonzero, so $S(x,y,z)=m^{\mathsf T}Gm>0$. Together with \cref{eq:P-definition}, this establishes \cref{eq:DeterminantInequality}.

\bibliography{Reference}

\end{document}